\documentclass[acmsmall,screen]{acmart}

\setcopyright{acmlicensed}
\acmYear{2021}
\copyrightyear{2021}
\acmSubmissionID{popl21main-p41-p}
\acmJournal{PACMPL}
\acmVolume{5}
\acmNumber{POPL}
\acmArticle{8}
\acmMonth{1}
\acmPrice{}
\acmDOI{10.1145/3434289}

\usepackage{booktabs}   
\usepackage{subcaption} 
\usepackage{tikz}
\usetikzlibrary{shapes.misc,arrows.meta}

\usepackage{mathtools}
\usepackage{amsmath}
\usepackage{amsthm}
\usepackage{enumitem}
\usepackage{framed}
\usepackage{nicefrac}
\usepackage{pifont}
\usepackage{algorithm2e}
\usepackage{hhline}
\usepackage{wrapfig}
\theoremstyle{plain}
\newtheorem{theorem}{Theorem}
\newtheorem{lemma}[theorem]{Lemma}
\newtheorem{corollary}[theorem]{Corollary}
\newtheorem{proposition}[theorem]{Proposition}
\newtheorem{claim}{Claim}
\theoremstyle{definition}
\newtheorem{definition}{Definition}
\newtheorem{example}{Example}

\theoremstyle{remark}
\newtheorem*{remark}{Remark}

\theoremstyle{remark}

\newcommand{\stat}{z}

\newcommand{\states}{Z}

\DeclarePairedDelimiter{\abs}{\lvert}{\rvert}
\newcommand{\integers}{\mathbb{Z}}
\newcommand{\Rats}{\mathbb{Q}}
\newcommand{\Reals}{\mathbb{R}}
\newcommand{\Nats}{\mathbb{N}}
 
 \newcommand{\sdom}{\mathsf{DOM}}
 \newcommand{\ptransf}{\Delta}
\newcommand{\powerset}[1]{\mathfrak{P}{(#1)}}
\newcommand{\Lap}[1]{\mathsf{Lap}{(#1)}}
\newcommand{\DLap}[1]{\mathsf{DLap}{(#1)}}

\newcommand{\pexp}[1]{\mathsf{Exp}{(#1)}}
\newcommand{\euler}{e}
\newcommand{\eulerv}[1]{\ensuremath{\euler^{#1}}}
\newcommand{\cU}{\mathcal{U}}
\newcommand{\cV}{\mathcal{V}}

\newcommand{\cQ}{\mathcal{Q}}

\newcommand{\Prob}{\mathsf{Prob}}

\newcommand{\st}  {\mathbin{|}}

\newcommand{\set}[1]{\{#1\}}
\newcommand{\norm}[1]{\mathbin{||}#1 \mathbin{||}}
\newcommand{\true}{\mathsf{true}}
\newcommand{\false}{\mathsf{false}}
\newcommand{\cA}{\mathcal{A}}

\newcommand{\cD}{\mathcal{D}}
\newcommand{\cL}{\mathcal{L}}
\newcommand{\cX}{\mathcal{X}}
\newcommand{\cB}{\mathcal{B}}
\newcommand{\cZ}{\mathcal{Z}}
\newcommand{\cR}{\mathcal{R}}
\newcommand{\sF}{\mathcal{F}}
\newcommand{\mx}{\mathsf{max}}
\newcommand{\bv}{\mathsf{b}}
\newcommand{\dv}{\mathsf{x}}
\newcommand{\iv}{\mathsf{z}}
\newcommand{\rv}{\mathsf{r}}
\newcommand{\mbar} {\mathbin{|}}
\newcommand{\Ifs}{\textsf{if}}
\newcommand{\Thens}{\textsf{then}}
\newcommand{\eds}{\textsf{end}}
\newcommand{\Elses}{\textsf{else}}
\newcommand{\Whiles}{\textsf{while}}
\newcommand{\Dos}{\textsf{do}}
\newcommand{\ifstatement}[3]{\Ifs \, #1\, \Thens\, #2 \, \Elses\, #3\, \eds}
\newcommand{\whilestatement}[2]{\Whiles \, #1\, \Dos\, #2\, \eds }

\newcommand{\ext}{\textsf{exit}}

\newcommand{\sLabels}{\mathsf{Labels}}

\newcommand{\chose}{\textsf{choose}}
\newcommand{\exit}{\textsf{exit}}

\newcommand{\distance}{distance to disagreement}

\DeclareMathOperator*{\argmax1}{arg\,max}

\newcommand{\rlin} {{\mathfrak R}_{+}}
\newcommand{\rpol} {{\mathfrak R}_{+,\times}}
\newcommand{\rexp} {{\mathfrak R}_{\mathsf{exp}}}
\newcommand{\lngreal} {\cL_{\mathsf{exp}}}
\newcommand{\threal} {\mathsf{Th}_{\mathsf{exp}}}
\newcommand{\fthreal} {\mathsf{Th}_{\mathsf{exp}}^{\mathsf{full}}}
\newcommand{\folreal} {\mathsf{Th}_{+}}
\newcommand{\foreal} {\mathsf{Th}_{+,\times}}

\newcommand{\tuple}[1] {\langle #1 \rangle}
\newcommand{\s}[1] {\mathsf{#1}}

\newcommand{\ourlang} {{\sf DiPWhile}}

\newcommand{\newlang} {{\sf DiPWhile+}}
\newcommand{\tool} {{\sf DiPC}}

\newcommand{\sem}[1] {[\![#1]\!]}

\newcommand{\rmv}[1] {}
\newcommand{\dom} {\mathsf{dom}}

\newcommand{\cmark}{\ding{51}}
\newcommand{\xmark}{\ding{55}}
\newcommand{\toolplus} {{\sf DiPC+}}
\makeatletter
\newcommand{\removelatexerror}{\let\@latex@error\@gobble}
\makeatother

\newcommand{\dd}{\mathsf{dd}}
\newcommand{\ddf}[2]{\dd_{#1,#2}}

\newcommand{\din}{d}
\newcommand{\dout}{d'}
 \newcommand{\psin}{\psi_{in}}
 \newcommand{\psiout}{\psi_{out}}
  \newcommand{\psipr}{\psi_{pr}}
  \newcommand{\phin}{\phi_{in}}
 \newcommand{\phiout}{\phi_{out}}
  \newcommand{\phipr}{\phi_{pr}}

\newcommand{\Iag} {I_{\alpha,\gamma}}
\newcommand{\Tag} {\theta_{\alpha,\gamma}}

\newcommand{\vrin}{\overline{\mathsf{vr}_{in}}}
\newcommand{\vrout}{\overline{\mathsf{vr}_{out}}}
\newcommand{\ovr}{\overline{\mathsf{vr}}}

\newcommand{\rvin}{\overline{\rv_{in}}}
\newcommand{\rvout}{\overline{\rv_{out}}}
\newcommand{\orv}{\overline{\rv}}

\newcommand{\reg}{\eta}
\newcommand{\regf}[2]{\eta (#1, #2)}

 \renewcommand{\det}{\mathsf{det}}

\newcommand{\svtone}{\textsf{SparseVariant}}
\newcommand{\svttwo}{\textsf{Sparse}}
\newcommand{\nsvt}{\textsf{NumericSparse}}
\newcommand{\athres}{\textsf{AboveThreshold}}
\newcommand{\nmax}{\textsf{NoisyMax}}
\newcommand{\laplace}{\textsf{Laplace}}
\newcommand{\exponential}{\textsf{Exponential}}

\makeatletter
\g@addto@macro\bfseries{\boldmath}
\makeatother

\newcommand*{\AppendixTrue}{}%

\begin{document}

	\title[Deciding Accuracy]{Deciding Accuracy of Differential Privacy Schemes}         


	\author{Gilles Barthe}
	\affiliation{
		\institution{Max Planck Institute for Security and Privacy, Bochum}            
		\country{Germany}                    
	}
	\email{gjbarthe@gmail.com}          

	\author{Rohit Chadha}
	\affiliation{
		\institution{University of Missouri}            
		\country{USA}                    
	}
	\email{chadhar@missouri.edu}          

	\author{Paul Krogmeier}
	\affiliation{
		\institution{University of Illinois, Urbana-Champaign}            
		\country{USA}                    
	}
	\email{paulmk2@illinois.edu}          

	\author{A.~Prasad Sistla}
	\affiliation{
		\institution{University of Illinois, Chicago}            
		\country{USA}                    
	}
	\email{sistla@uic.edu}          

	\author{Mahesh Viswanathan}
	\affiliation{
		\institution{University of Illinois, Urbana Champaign}            
		\country{USA}                    
	}
	\email{vmahesh@illinois.edu}          


\begin{abstract}
  Differential privacy is a mathematical framework for developing
  statistical computations with provable guarantees of privacy and
  accuracy. In contrast to the privacy component of differential
  privacy, which has a clear mathematical and intuitive meaning, the
  accuracy component of differential privacy does not have a general
  accepted definition; accuracy claims of differential privacy
  algorithms vary from algorithm to algorithm and are not
  instantiations of a general definition. We identify program
  discontinuity as a common theme in existing \emph{ad hoc} definitions and
  introduce an alternative notion of accuracy parametrized by, what we
  call, {\distance} --- the {\distance} of an input $x$ w.r.t.\, a
  deterministic computation $f$ and a distance $d$, is the minimal
  distance $d(x,y)$ over all $y$ such that $f(y)\neq f(x)$.  We show
  that our notion of accuracy subsumes the definition used in
  theoretical computer science, and captures known accuracy claims
  for differential privacy algorithms. In fact, our general notion of
  accuracy helps us prove better claims in some cases. Next, we study
  the decidability of accuracy. We first show that accuracy is in general
  undecidable. Then, we define a non-trivial class of probabilistic
  computations for which accuracy is decidable (unconditionally, or
  assuming Schanuel's conjecture).  We implement our decision
  procedure and experimentally evaluate the effectiveness of our
  approach for generating proofs or counterexamples of accuracy for
  common algorithms from the literature.
\end{abstract}




	\maketitle

	\section{Introduction}
\label{sec:intro}

Differential privacy~\cite{DMNS06,DR14} is a mathematical framework
for performing privacy-preserving computations over sensitive
data. One important feature of differential privacy algorithms is their ability to achieve
provable individual privacy guarantees and at the same time ensure that the outputs are reasonably accurate.
In the case of privacy, these guarantees relate executions of the
differentially private algorithm on adjacent databases. These privacy
guarantees are an instance of relational properties, and they have
been extensively studied in the context of program
verification~\cite{RP10,BKOZ13,GHHNP13,ZK17,AH18,BartheCJS020} and program
testing~\cite{BichselGDTV18,DingWWZK18}. In the case of accuracy,
these guarantees relate the execution of the differentially private
algorithm to that of an \lq\lq ideal\rq\rq\ algorithm, which can be
assumed to be deterministic. Typically, \lq\lq ideal\rq\rq\ algorithms
compute the true value of a statistical computation, while
differentially private algorithms compute noisy versions of the right
answer. The precise relationship between the output of the ``ideal''
algorithm and the differentially private one varies from algorithm to
algorithm, and no general definition of accuracy has been proposed in
this context. In some cases, the definition of accuracy is similar to
the one used in the context of randomized algorithm
design~\cite{motwani-raghavan}, where we require that the output of
the differentially private algorithm be close to the true output (say
within distance $\gamma$) with high probability (say at least
$1-\beta$). Such a notion of accuracy is similar to computing error
bounds, and there has been work on formally verifying such a
definition of accuracy on some
examples~\cite{BGGHS16-icalp,SmithHA19,LoboRG20}.  Unfortunately,
prior work on formal verification of accuracy suffers from two
shortcomings:
\begin{itemize}
\item the lack of a general definition for accuracy: as pointed out,
  each differential privacy algorithm in the literature has its own
  specific accuracy claim that is not an instantiation of a
  general definition. The lack of a general definition for accuracy
  means that the computational problem of ``verifying accuracy'' has not
  been defined, which prevents its systematic study.

\item imprecise bounds: accuracy bounds in theoretical papers are often established using
  concentration bounds, e.g.\, Chernoff bound, by hand. Existing verification frameworks, with the exception
  of \cite{LoboRG20}, cannot be used  to verify accuracy
  claims that are established using concentration bounds. This is due to the fact that applying concentration bounds generally requires proving independence, which
  is challenging. Thus, with the exception
  of \cite{LoboRG20}, usually only weaker versions of accuracy claims can be verified automatically by the existing techniques.
\end{itemize}
This paper overcomes both shortcomings by proposing a general notion of
accuracy and by proving decidability of accuracy for a large class of
algorithms that includes many differentially private algorithms from
the literature.

\subsubsection*{Technical contributions}
Our focus is the verification of accuracy claims for differential privacy algorithms that aim to bound the error of getting the \emph{correct} answer. 
Other notions of utility or accuracy, such as those that depend on variance and other moments are out of the scope of this paper. Our contributions in this space are three-fold: (a) we give a definition of accuracy that captures accuracy claims known in the literature, (b) study the computational problem of checking accuracy, as identified by our definition of accuracy, and (c) perform an experimental evaluation of our decision procedure.

\paragraph{\textbf{General definition of accuracy.}} The starting point of our work is the (well-known) observation that the usual
definition used in theoretical computer science to measure the utility
of a randomized algorithm (informally discussed above) fails to
adequately capture the accuracy claims made for many differential
privacy algorithms. To see why this is the case, consider {\svttwo}
(also called Sparse Vector Mechanism or SVT). The problem solved by
{\svttwo} is the following: given a threshold $T$, a database $x$, and
a list of queries of length $m$, output the list of \emph{indices} of
the first $c$ queries whose output on $x$ is greater than or equal to
$T$. {\svttwo} solves this problem while maintaining the privacy of
the database $x$ by introducing noise to the query answers as well as
to $T$ when comparing them. Because of this, {\svttwo}'s answers
cannot be accurate with high probability, if the answers to the query
are very close to $T$ --- the correct answer is ``discontinuous''
near $T$ while the probability distributions of the introduced noise are continuous functions. Thus, we can expect {\svttwo} to be accurate only when all
query answers are bounded away from $T$. This is what the known
accuracy claim in the literature proves.

The first contribution of
this paper is a more general notion of accuracy that takes into  account  the
\emph{\distance} w.r.t.\, the \lq\lq ideal\rq\rq\ algorithm that is
necessary in accuracy claims for differential privacy
algorithms. Informally, an input $u$ of a deterministic computation
$f$ has {\distance} $\alpha$ (with respect to a distance function $d$
on the input space) if $\alpha$ is the largest number such that whenever $f(u)\ne f(v)$ then $d(u,v)\geq \alpha.$ 
This notion of {\distance} is inspired from the
Propose-Test-Release mechanism~\cite{DR14}, but we use it for the
purpose of accuracy rather than privacy.  More precisely, our notion
of accuracy requires that for every input $u$ whose {\distance} is
greater than
 $\alpha$, the output of the differentially private algorithm be
within a distance of $\gamma$ from the correct output, with probability
at least $1-\beta$. The traditional accuracy definition, used in the literature
on randomized algorithms, is obtained by setting $\alpha=0$. Our definition
captures most \emph{ad hoc} accuracy claims known in the literature
for different differential privacy algorithms. Our definition is reminiscent of
accuracy definitions that use three parameters in \citet{BlumLR13} and \citet{BhaskarLST10}. However, neither 
of these definitions have a notion like {\distance}~\footnote{A more detailed comparison with these definitions can be found in Section~\ref{sec:accuracy-def}.}.

We show that the additional degree of flexibility in our definition of accuracy with the introduction of parameter $\alpha$, can be exploited to improve known accuracy bounds for {\nsvt}, a
variant of {\svttwo} that returns the noised query answers when they are above a threshold. Specifically, the accuracy bound
from~\citet{DR14} translates in our framework to
$(\alpha, \beta,\alpha)$-accuracy for all $\alpha$, and for
$\beta=\beta_0(\alpha)$ for some function $\beta_0$. In contrast, we
can prove $(\alpha,\beta,\gamma)$-accuracy for all $\alpha,\gamma$ and
$\beta=\beta_1(\alpha,\gamma)$. Our result is more general and more
precise, since $\beta_1(\alpha,\alpha)$ is approximately
$\frac{1}{2}\beta_0 (\alpha)$ for all values of $\alpha$.

\paragraph{\textbf{Deciding Accuracy.}}
Establishing a general definition of accuracy allows us to study the  
decidability of the problem of checking
accuracy. Differential privacy
algorithms are typically parametrized by the privacy budget
$\epsilon$, where program variables are typically sampled from
distributions whose parameters depend on $\epsilon$. Thus, verifying a
property for a differential privacy algorithm is to verify an
\emph{infinite family} of programs, obtained by instantiating the
privacy budget $\epsilon$ to different values. We, therefore, have two
parametrized verification problems, which we respectively call the
\emph{single-input} and \emph{all-inputs} problems. These problems state: given a
parametrized program $P_\epsilon$, an interval $I$ and accuracy bounds
$(\alpha,\beta,\gamma)$ that may depend on $\epsilon$, is $P_\epsilon$
$(\alpha,\beta, \gamma)$-accurate at input $u$ (resp. at all inputs)
for all possible values of $\epsilon\in I$?

We first show that accuracy is in general undecidable, both for the
single-input and all-inputs variants. Therefore, we focus on
decidability for some specific class of programs. We follow the
approach from~\citet{BartheCJS020}, where the authors propose a
decision procedure for $(\epsilon,\delta)$-differential privacy of a
non-trivial class of (parametric) programs, {\ourlang}, with a finite number of inputs and output variables taking values in a finite domain. Specifically, we carve out
a class of programs, called {\newlang}, whose operational semantics
have a clever encoding as a finite state discrete-time Markov
chain. Then, we use this encoding to reduce the problem of accuracy to
the theory of reals with exponentials. Our class of programs is larger
than the class of programs {\ourlang}; it
supports the use of a finite number of real input and real output variables and permits the
use of real variables as means of Laplace distributions when sampling
values. 

We show that checking accuracy for both \emph{single-input} and \emph{all-inputs} is decidable for {\newlang} programs including those with
real input and output variables, assuming Schanuel's
conjecture. Schanuel's conjecture is a long-standing open problem in
transcendental number theory, with deep applications in several areas
of mathematics. In our proof, we use a celebrated result of~\citet{mw96}, which shows that Schanuel's conjecture entails
decidability of the theory of reals with exponentials. Our decidability proof essentially encodes the \emph{exact} probability of the algorithm yielding an
output that is $\gamma$ away from the correct answer. As we calculate exact probabilities, we do not have to resort to concentration bounds for verifying accuracy.

We also identify sufficient conditions under which the \emph{single-input} problem is decidable for {\newlang} programs \emph{unconditionally}, i.e., without assuming Schanuel's conjecture. These unconditional decidability results rely on two crucial observations. 
First, to check accuracy at specified input $u$, it suffices to set $\alpha$ to the distance to disagreement for $u$. 
This is because $\beta$ decreases when $\alpha$ increases. 
Secondly, given a {\newlang} program $P$ and real number $\gamma,$ 
we can often construct a new {\newlang} program $P^{\mathsf{new}}$ such that  $P^{\mathsf{new}}$ outputs true on input $u$ if and only if the output produced by $P$
on $u$ is at most $\gamma$ away from the correct answer. This allows us only to consider the programs that produce outputs from a finite domain. The single-input accuracy problem can then be 
expressed in \citet{mccallum2012deciding}'s decidable
fragment of the theory of reals with
exponentials.

An immediate consequence of our unconditional decidability results is that the \emph{all-inputs} accuracy problem is decidable for programs with inputs and
outputs that come from a finite domain. This is essentially the class of programs {\ourlang}, for which differential privacy was shown to be decidable in \citet{BartheCJS020}. Our decidability results are summarized in Table~\ref{tab:test} on Page~\pageref{tab:test}.

\paragraph{\textbf{Experimental Evaluation.}} 
We adapt the {\tool} tool
from \citet{BartheCJS020} to verify accuracy bounds at given inputs and evaluate it on
many examples from the literature. Our adaptation takes as input a
program in {\newlang}, constructs a sentence in the decidable
\citet{mccallum2012deciding} fragment of the theory of reals with
exponentials and calls Mathematica\textregistered~to see if the
sentence is valid. Using the tool, we verified the accuracy of $\svttwo$, $\nmax$,
Laplace Mechanism, and $\nsvt$ at specified inputs. Our tool also found counter-examples for ${\svttwo}$ and ${\svtone}$, when accuracy claims do not hold.
In addition, our tool is able to verify
improvements of accuracy bounds for $\nmax$ over known accuracy bounds in the literature  given in this paper. 
Finally, we experimentally found better potential accuracy bounds for our examples by running our tool on progressively smaller $\beta$ values.   
  
\ifdefined\AppendixTrue
This is the author's version of the paper. It is posted here for your personal use. Not for redistribution. The definitive version was published in the Proceedings of the ACM on Programming Languages (POPL' 2021), available at https://doi.org/10.1145/3434289.
\else 
Due to lack of space, some proofs and other materials have
been omitted. The omitted material can be located in the
arXiv repository \cite{POPLarxiv}.
\fi


        \section{Definition of Accuracy}
\label{sec:accuracy-def}

In the differential privacy model~\cite{DMNS06}, a trusted curator
with access to a database returns answers to queries made by possibly
dishonest data analysts that do not have access to the database. The
task of the curator is to return probabilistically noised answers so
that data analysts cannot distinguish between two databases that are
adjacent, i.e.\, only differ in the value of a single
individual. However, an overriding concern is that, in spite of the
noise, responses should still be sufficiently close to the actual
answers to ensure the usefulness of any statistics computed on the
basis of those responses. This concern suggests a requirement for
accuracy, which is the focus of this paper. The definition of
accuracy, one of the main contributions of this paper, is presented
here.

We start by considering the usual definition used in theoretical
computer science~\cite{motwani-raghavan} to characterize the quality
of a randomized algorithm $P$ that approximately computes a function
$f$. Informally, such a definition demands that, for any input $x$,
the output $P(x)$ be ``close'' to function value $f(x)$ with ``high
probability''. In these cases, ``close'' and ``high probability'' are
characterized by parameters (say) $\gamma$ and $\beta$. Unfortunately,
such a definition is too demanding, and is typically not satisfied by
differential privacy algorithms. We illustrate this with the following
example.
\begin{example}
	\label{ex:accuracy-def}
	Consider {\svttwo} (also called Sparse Vector Technique or SVT). The
  problem solved by {\svttwo} is the following: given a threshold $T$,
  a database $x$ and a list of queries of length $m$, output the list
  of \emph{indices} of the first $c$ queries whose output on $x$ is
  above $T$. Since the goal of {\svttwo} is to maintain privacy of the
  database $x$, {\svttwo} introduces some small noise to the query
  answers as well as to $T$ before comparing them, and outputs the
  result of the ``noisy'' comparison; the exact pseudocode is given in
  Figure~\ref{algo-svt2}. Observe that if the answers to queries are
  very close to $T$, then {\svttwo}'s answer will not be ``close'' to
  the right answer (for non-noisy comparison) since the noisy
  comparison could give any result. On the other hand, if the query
  answer is far from $T$, then the addition of noise is unlikely to
  change the result of the comparison, and {\svttwo} is likely to give
  the right answer with high probability (despite the noise). Thus,
  {\svttwo}'s accuracy claims in the literature do not apply to all
  inputs, but only to those databases and queries whose answers are
  far away from the threshold $T$.
\end{example}
This example illustrates that the accuracy claims in differential
privacy can only be expected to hold for inputs that are far away from
other inputs that \emph{disagree} --- in Example~\ref{ex:accuracy-def}
above, accuracy claims don't hold when query answers are close to the
threshold because these inputs are close to other inputs in which the
comparison with the threshold will give the opposite result. This
problem motivates our general definition of accuracy, which captures
all accuracy claims for several differential privacy
algorithms. Before presenting the formalization, we introduce some
notation and preliminaries that will be useful.

\paragraph{Notation.}
We denote the set of real numbers, rational numbers, natural numbers,
and integers by $\Reals,\Rats,\Nats$, and $\integers$, respectively.
We also use
$\Reals^\infty\:= \Reals \cup \set{\infty},
\Reals^{>0}=\set{x\in\Reals\st x>0},\Reals^{\geq 0}\:=\set{x\in
  \Reals\st x\geq 0}$. The Euler constant is denoted by $\euler$. For
any $m>0$ and any vector $a=(a_1,...,a_m)\in\integers^m$ (or,
$a\in \Reals^m$, $a\in\Rats^m$, $a\in \Nats^m$), recall that
$\norm{a}_1\:=\sum_{1\leq i\leq m}\mathbin{|}a_i\mathbin{|}$ and
$\norm{a}_\infty\:=\max\{\mathbin{|}a_i\mathbin{|}\st 1\leq i\leq
m\}$.

A differential privacy algorithm is typically a probabilistic program
whose behavior depends on the privacy budget $\epsilon$. When we
choose to highlight this dependency we denote such programs as
$P_\epsilon$, and when we choose to ignore it, e.g. when $\epsilon$
has been fixed to a particular value, we denote them as just
$P$. Since $P$ is a probabilistic program, it defines a randomized
function $\sem{P}$, i.e., on an input $u \in \cU$ (where $\cU$ is the
set of inputs), $\sem{P}(u)$ is a distribution on the set of outputs
$\cV$. We will often abuse notation and use $P(u)$ when we mean
$\sem{P}(u)$, to reduce notational overhead. For a measurable set
$S \subseteq \cV$, the probability that $P$ outputs a value in $S$ on
input $u$ will be denoted as $\Prob(P(u) \in S)$; when $S$ is a
singleton set $\set{v}$ we write $\Prob(P(u) = v)$ instead of
$\Prob(P(u) \in \set{v})$.

The accuracy of a differential privacy algorithm $P$ is defined with
respect to an ``ideal'' algorithm that defines the function that $P$
is attempting to compute while maintaining privacy. We denote this
``ideal'' function as $\det(P)$. It is a deterministic function, and
so $\det(P) : \cU \to \cV$.

To define accuracy, we need a measure for when inputs/outputs are
close. On inputs, the function measuring closeness does not need to be
a metric in the formal sense. We assume
$d:\cU\times \cU\rightarrow \Reals^\infty$ is a ``distance'' function
defined on $\cU$, satisfying the following properties: for all
$u,u'\in \cU$, $d(u,u')=d(u',u)\geq 0$, $d(u,u)\:=0.$ For any
$X\subseteq \cU$ and any $u\in \cU$, we let
$d(u,X)= \inf_{u'\in X} d(u,u')$. On outputs, we will need the
distance function to be dependent on input values. Thus, for every
$u\in \cU$, we also assume that there is a distance function $d'_u$
defined on $\cV$. For any $v\in \cV,u\in \cU$ and for any
$\gamma\in \Reals^{\geq 0}$, let
$B(v,u,\gamma)=\set{v'\in \cV\st d'_u(v,v')\leq \gamma}$. In other
words, $B(v,u,\gamma)$ is a ball of radius $\gamma$ around $v$ defined
by the function $d'_u$. We next introduce a key notion that we call
\emph{\distance}.
\begin{definition}
	\label{def:disagreement}
	For a randomized algorithm $P$ and input $u \in \cU$, the
  \emph{\distance} of $P$ and $u$ with respect to $\det(P)$ is the
  minimum of the distance between $u$ and another input $u'$ such that
  the outputs of $\det(P)$ on $u$ and $u'$ differ. This can defined
  precisely as
	\[
	\dd(P,u) = d(u, \{\cU - \det(P)^{-1}(\det(P)(u))\}).
	\]
\end{definition}

We now have all the components we need to present our definition of
accuracy. Intuitively, the definition says that a differential privacy
algorithm $P$ is accurate with respect to $\det(P)$ if on all inputs
$u$ that have a large {\distance} (as measured by parameter $\alpha$),
$P$'s output on $u$ is close (as measured by parameter $\gamma$) to
$\det(P)(u)$ with high probability (as measured by parameter
$\beta$). This is formalized below.
\begin{definition}
\label{def:accuracy}
Let $\alpha,\beta,\gamma\in \Reals^{\geq 0}$ such that
$\beta\in [0,1]$. Let $P$ be a differential privacy algorithm on
inputs $\cU$ and outputs $\cV$. $P$ is said to be
$(\alpha,\beta,\gamma)$-accurate at input $u\in \cU$ if the following
condition holds: if $\dd(P,u)>\alpha$ then
$\Prob(P(u)\in B(\det(P)(u),u,\gamma))\geq 1-\beta.$

We say that $P$ is $(\alpha,\beta,\gamma)$-accurate if for all
$u\in \cU$, $P$ is $(\alpha,\beta,\gamma)$-accurate at $u.$
\end{definition}

Observe that when $\alpha = 0$, $(\alpha,\beta,\gamma)$-accuracy
reduces to the standard definition used to measure the precision of a
randomized algorithm that approximately computes a
function~\cite{motwani-raghavan}. All three parameters $\alpha,\beta$,
and $\gamma$ play a critical role in capturing the accuracy claims
known in the literature. As we will see in
Section~\ref{sec:accuracy-ex}, we show that the \textsf{Laplace} and
\textsf{Exponential} Mechanisms are $(0,\beta,\gamma)$-accurate,
{\athres} and {\svttwo} are $(\alpha,\beta,0)$-accurate, and {\nsvt}
is $(\alpha,\beta,\gamma)$-accurate. Finally, observe that as $\alpha, \gamma$ increase the error probability $\beta$ decreases.

\paragraph*{Comparison with alternative definitions} As previously noted, it is well-known that the usual definition of accuracy from randomized algorithms does not capture desirable notions of accuracy for differentially private computations, and a number of classic papers from the differential privacy literature have proposed generalizations of the usual notion of accuracy with a third parameter. For instance, \citet{BlumLR13} introduce a relaxed notion of accuracy in order to study lower bounds; their definition is specialized to mechanisms on databases, and given with respect to a class $\mathcal{C}$ of (numerical) queries. Informally, a mechanism $A$ is accurate if for every query $Q$ in the class $\mathcal{C}$ and database $D$, there exists a nearby query $Q'$ such that with high probability the output $Q(D)$ is close to $Q'(D)$. Their definition is of a very different flavour, and is not comparable to ours. Another generalization is given in \citet{BhaskarLST10}, for algorithms that compute frequent items. These algorithms take as input a list of items and return a list of most frequent items and frequencies. Their notion of usefulness requires that with high probability the frequencies are close to the true frequency globally, and for each possible list of items. None of these definitions involve a notion like {\distance}. 
        \section{Examples}
\label{sec:accuracy-ex}

The definition of accuracy (Definition~\ref{def:accuracy}) is general
enough to capture all accuracy claims we know of in the
literature. It's full generality seems to be needed in order to
capture known results. In this section, we illustrate this by looking
at various differential privacy mechanisms and their accuracy
claims. As a byproduct of this investigation, we also obtain tighter
and better bounds for the accuracy of {\nsvt}.

\subsection{{\laplace} Mechanism}
\label{sec:laplace}
The {\laplace} mechanism~\cite{DMNS06} is the simplest differential
privacy algorithm that tries to compute, in a privacy preserving
manner, a numerical function $f: \cU \to \cV$, where $\cU\:=\Nats^n$
and $\cV\:=\Reals^k$, where $k > 0$. The algorithm adds noise sampled
from the Laplace distribution. Let us begin by defining this
distribution.
\begin{definition}[Laplace Distribution]
	\label{def:laplace}
	Given $\epsilon>0$ and mean $\mu$, let $\Lap{\epsilon,\mu}$ be the continuous distribution whose probability density function (p.d.f.) is given by
	\[
	f_{\epsilon,\mu}(x) = \frac {\epsilon} 2 \ \euler^{-\epsilon \abs {x -\mu}}.
	\]
	$\Lap{\epsilon,\mu}$ is said to be the \emph{Laplace distribution} with mean $\mu$ and scale parameter $\frac{1}{\epsilon}$.

	It is sometimes useful to also look at the discrete version of the above distribution. Given $\epsilon>0$ and mean $\mu,$ let $\DLap{\epsilon,\mu}$ be the discrete distribution on $\integers$, whose probability mass function (p.m.f.) is
	\[
	f_{\epsilon,\mu} (i) = \frac {1-\euler^{-\epsilon}} {1+\euler^{-\epsilon}} \ \euler^{-\epsilon \abs {i -\mu}}.
	\]
	$\DLap{\epsilon,\mu}$ is said to be the \emph{discrete Laplace distribution} with mean $\mu$ and scale parameter $\frac 1 \epsilon$.
\end{definition}
On an input $u \in \cU$, instead of outputting $f(u)$, the {\laplace}
mechanism ($P^{\s{Lap}}_\epsilon$) outputs the value
$f(x) + (Y_1,\ldots Y_k)$, where each $Y_i$ is an independent,
identically distributed random variable from
$\Lap{(\frac{\epsilon}{\Delta f},0)}$; here $\Delta f$ is the
\emph{sensitivity} of $f$, which measures how $f$'s output changes as
the input changes~\cite{DR14}.

Theorem 3.8 of \citet{DR14} establishes the following accuracy claim for {\laplace}.
\begin{theorem}[Theorem 3.8 of \citet{DR14}]
	\label{thm:laplace}
	For any $u\in \cU$, and $\delta\in (0,1]$
\[
Pr \left[\norm{f(u)-P^{\s{Lap}}_\epsilon(u)}_\infty \geq \ln
\left(\frac{k}{\delta}\right)\left(\frac{\Delta f}{\epsilon}\right)\right] \leq \delta
\]
\end{theorem}

We can see that Theorem~\ref{thm:laplace} can be rephrased as an
accuracy claim using our definition. Observe that here
$\det(P^{\s{Lap}}_\epsilon)\:=f$. Let the distance function $d$ on
$\cU$ be defined by $d(u,u')=0$ if $u=u'$ and $d(u,u')=1$
otherwise. For $u\in \cU$, let the distance function $d'_u$ on $\cV$
be defined by $d'_u(v,v')\:=\norm{v-v'}_\infty$. Now, it is easily
seen that the above theorem is equivalent to stating that the
{\laplace} mechanism is $(0,\delta,\gamma)$-accurate for all
$\epsilon,\gamma$, where
$\delta=k \euler^{-\frac{\gamma \epsilon}{\Delta f}}$.

\subsection{{\exponential} Mechanism}
\label{sec:exponential}

Consider the input space $\cU = \Nats^n$. Suppose for an input
$u \in \cU$, our goal is to output a value in a finite set $\cV$ that
is the ``best'' output. Of course for this to be a well-defined
problem, we need to define what we mean by the ``best'' output. Let us
assume that we are given a utility function
$F: \cU \times \cV \to \Reals$ that measures the quality of the
output. Thus, our goal on input $u$ is to output
$\argmax1_{v \in \cV} F(u,v)$~\footnote{If there are multiple $v$ that
  maximize the utility, there is a deterministic criterion that
  disambiguates.}.

The {\exponential} mechanism~\cite{McSherryT07}
($P^{\s{Exp}}_\epsilon$) solves this problem while guaranteeing
privacy by sampling a value in $\cV$ based on
the \emph{exponential distribution}. This distribution depends on the utility function $F$ and is defined below.
\begin{definition}[Exponential Distribution]
	\label{def:exponential}
	Given $\epsilon>0$ and $u\in \cU$, the discrete distribution $\pexp {\epsilon, F,u}$ on $\cV$ is given by the probability mass function:
	\[
	h_{\epsilon, F,u} (v)=\frac{\euler^{\epsilon F(u,v)}}{ \sum_{v\in \cV} \euler^{\epsilon F(u,v)}}.
	\]
\end{definition}
On an input $u \in \cU$, the {\exponential} mechanism outputs
$v \in \cV$ according to distribution
$\pexp{\frac{\epsilon}{\Delta F},F,u}$, where $\Delta F$ is the
sensitivity of $F$. Taking $\det(P^{\s{Exp}}_\epsilon)$ to be the
function such that $\det(P^{\s{Exp}}_\epsilon)(u)= \argmax1_{v \in \cV} F(u,v)$, the following claim about the {\exponential}
mechanism is proved in Corollary 3.12 of \citet{DR14}.
\begin{theorem}[Corollary 3.12 of \citet{DR14}]
	\label{thm:exponential}
	For any $u$ and any $t>0$,
\[
Pr \left[ F(u,P^{\s{Exp}}_\epsilon(u))\leq F(u,\det(P^{\s{Exp}}_\epsilon)(u))- \frac{2\Delta F}{\epsilon} (\ln (\mathbin{|}\cV
\mathbin{|}) + t) \right] \leq \euler^{-t}
\]
\end{theorem}

Again we can see Theorem~\ref{thm:exponential} as an accuracy claim by
our definition. Let the distance function $d$ on $\cU$ be defined by
$d(u,u')=0$ if $u=u'$ and $d(u,u')=1$ otherwise. For any $u \in \cU$,
take the distance metric $d'_u$ to be
$d'_u(v,v') = \mathbin{|}F(u,v) - F(u,v')\mathbin{|}$, for
$v,v'\in \cV$. Theorem~\ref{thm:exponential} can be seen as saying
that, for all $\epsilon$ and $t$, the {\exponential} mechanism is
$(0,\beta,\gamma)$-accurate where $\beta = \euler^{-t}$,
$\gamma=\frac{2\Delta F}{\epsilon} (\ln \mathbin{|}\cV \mathbin{|} +
\ln (\frac{1}{\beta}))$.

\subsection{{\nmax}}
\label{sec:nmax}

Consider the following problem. Given a sequence
$(q_1,q_2,\ldots, q_m)$ of elements (with each $q_i \in \Reals$),
output the smallest index of an element whose value is the maximum in
the sequence. The algorithm {\nmax} is a differentially private
way to solve this problem. It is shown in
Figure~\ref{fig:nmax}.
\begin{wrapfigure}{r}{0.27\textwidth}
	\RestyleAlgo{boxed}
	\removelatexerror
	\begin{algorithm}[H]
		\DontPrintSemicolon
		\KwIn{$q[1:m]$}
		\KwOut{$out$}
		\;
		NoisyVector $\gets []$\;
		\For{$i\gets 1$ \KwTo $m$}{
			NoisyVector[i] $\gets$ $\Lap{\frac{\epsilon}{2}, q[i]}$
		}
		out $\gets$ argmax(NoisyVector)\;
	\end{algorithm}
	\caption{\footnotesize{Algorithm {\nmax}}}
	\label{fig:nmax}
\end{wrapfigure} Based on the privacy budget $\epsilon$, it
independently adds noise distributed according to
$\Lap{\frac{\epsilon}{2},0}$ and then outputs the index with the
maximum value after adding the noise.

Let us denote the deterministic function that outputs the index of the
maximum value in the sequence $(q_1,\ldots, q_m)$ by
$\det(\mbox{\nmax})$. On a given input sequence of length $m$, suppose
$i$ is the index output by {\nmax} and $j$ is the index output by
$\det(\mbox{\nmax})$. Theorem 6 of \citet{BGGHS16-icalp} proves that,
for any $\beta\in (0,1]$,
$Pr (q_j-q_i < \frac{4}{\epsilon}\ln\frac{m}{\beta})\geq 1-\beta.$

There are two different ways we can formulate NoisyMax in our
framework. In both approaches $\cU\:=\Reals^m.$ In the first
approach, the distance function $d$ on $\cU$ is the same as the one
given for the {\laplace} mechanism, i.e., for any $u,u'\in \cU$,
$d(u,u')=0$ if $u = u'$ and $d(u,u')=1$ otherwise. The set
$\cV\:=\set{i\::1\leq i\leq m}.$ For any $u=(q_1,...,q_m)\in \cU$, the
distance function $d'_u$ on $\cV$ is defined by
$d'_u(i,j)\:=|q_i-q_j|.$ Now, it is easy to see that the above
mentioned result of \citet{BGGHS16-icalp} is equivalent to the
statement that {\nmax} is $(0,\beta,\gamma)$-accurate where
$\beta\:=m\euler^{-\frac{\gamma\epsilon}{4}}.$

In the second approach, we use the distance functions $d$ on $\cU$ and
$d'_u$ on $\cV$ (for $u\in \cU$) defined as follows: for
$u,u'\in \cU$, $d(u,u')=\norm{u-u'}_\infty$, and for $u\in \cU$,
$i,j\in \cV$, if $q_i=q_j$ then $d'_u(i,j)=0$, otherwise
$d'_u(i,j)=1.$ We have the following lemma for the accuracy of
{\nmax}. 
\ifdefined\AppendixTrue
(See Appendix~\ref{sec:acc-ex-proofs} for the proof.)
\fi
\begin{lemma}
	\label{lem:nmax}
{\nmax} is $(\alpha,\beta,0)$-accurate for
$\beta\:=m\euler^{-\frac{\alpha\epsilon}{2}}$ and for all $\alpha\geq 0.$
\end{lemma}


\subsection{{\athres}}
\label{sec:at}

Given a sequence of queries $(q_1,\ldots, q_m)$ ($q_i \in \Reals$) and
parameter $T \in \Reals$, consider the problem of determining the
first query in the sequence which is \emph{above the threshold}
$T$. The goal is not to output the index of the query, but instead to
output a sequence of $\bot$ as long as the queries are below $T$, and
to terminate when either all queries have been read, or when the first
query $\geq T$ is read; if such a query is found, the algorithm outputs
$\top$ and stops.

{\athres} is a differentially private algorithm that solves the above
problem. The algorithm is a special case of {\svttwo} shown in
Figure~\ref{algo-svt2} when $c=1$. {\athres} works by adding noise to
$T$ and to each query, and comparing if the noised queries are
below the noised threshold. The noise added is sampled from the
Laplace distribution with scale parameter $\frac{2}{\epsilon}$ (for the threshold) and $\frac{4}{\epsilon}$ (for queries). The set of
inputs for {\athres} is $\cU = \Reals^m$ and the outputs are
$\cV = \set{\bot^m} \cup \set{\bot^k\top\: |\: k < m}$. The accuracy claims for {\athres}
in \citet{DR14} are given in terms of a notion of
\emph{$(\alpha,\beta)$-correctness}. Let $\cA$ be a randomized
algorithm with inputs in $\cU$ and outputs in $\cV$, and let
$\alpha \in \Reals^{\geq 0}$ and $\beta \in [0,1]$. We say that $\cA$
is $(\alpha,\beta)$-correct if for every $k,\:1\leq k\leq m$, and for
every input $u=(q_1,...,q_m)\in \cU$ such that $q_k\geq T+\alpha$ and
$q_i<T-\alpha$ for $1\leq i<k$, $\cA$ outputs $\bot^{k-1}\top$ with
probability $\geq 1-\beta.$ Using the results in \citet{DR14}, one can
prove the following lemma.
\begin{lemma}[\citet{DR14}]
\label{lem:AT}
{\athres} is $(\alpha,\beta)$-correct for  $\beta=
2m\euler^{-\frac{\alpha\epsilon}{8}}$ and for all $\alpha\geq 0.$
\end{lemma}

We formulate {\athres} in our framework and relate
$(\alpha,\beta)$-correctness to $(\alpha,\beta,0)$-accuracy as
follows.  Let the distance function $d$ on $\cU$ be given by
$d(u,u')\:=\norm{(u-u')}_\infty.$ The distance function $d'_u$ on
$\cV$ is given by $d'_u(v,v')=0$ if $v=v'$, otherwise $d'_u(v,v')=1.$
Now, we have the following lemma.

\begin{lemma}
	\label{lem:correct-acc}
  For any randomized algorithm $\cA$ as specified above, for any
  $\alpha,\beta$ such that $\alpha\geq 0$ and $\beta\in [0,1]$, $\cA$
  is $(\alpha,\beta)$-correct iff $\cA$ is
  $(\alpha,\beta,0)$-accurate.
\end{lemma}
\begin{proof}
  Let $\alpha,\beta$ be as given in the statement of the lemma. Now,
  assume $\cA$ is $(\alpha,\beta)$-correct. We show that it is
  $(\alpha,\beta,0)$-accurate. Let $v= \bot^{k-1}\top$ such that
  $1\leq k\leq m$. Now consider any
  $u=(u_1,...,u_m) \in \det(\cA)^{-1}(v)$ such that
  $d(u, \cU-(\det(\cA))^{-1}(v)))>\alpha.$ Now, if $k>1$ then fix any
  $i, i<k$ and consider $w=(w_1,...,w_m) \in \cU$ such that $w_i=T$
  and $w_j =u_j$ for all $j\leq m$ and $j\neq i.$ Clearly
  $\det(\cA)(w)\neq v$ and hence $w\in \cU-(\det(\cA))^{-1}(v).$ Since
  $d(u,w)>\alpha$, it is the case that $u_i< T-\alpha.$ Now, let
  $\delta\in \Reals$ such that $\delta>0$. Consider
  $w=(w_1,...,w_m)\in \cU$ such that $w_k= T-\delta$ and $w_i=u_i$ for
  $i\neq k, 1\leq i \leq m.$ Clearly,
  $w\in \cU - (\det(\cA))^{-1}(v))$, and hence $d(u,w) >\alpha$. Now,
  we have $u_k-w_k>\alpha$ and hence $u_k>T-\delta+\alpha.$ Since the
  last inequality holds for any $\delta>0$, we see that
  $u_k\geq T+\alpha.$ Thus, we see that $u_i<T-\alpha$ for $i<k$ and
  $u_k\geq T+\alpha.$ Since, $\cA$ is $(\alpha,\beta)$-correct, we see
  that it outputs $v$ with probability $\geq 1-\beta.$ Hence $\cA$ is
  $(\alpha,\beta,0)$-accurate.

  Now, assume that $\cA$ is $(\alpha,\beta,0)$-accurate. We show that
  it is $(\alpha,\beta)$-correct. Consider any
  $u=(u_1,...,u_m)\in \cU$ such that, for some $k\leq m$,
  $u_k\geq T+\alpha$ and for all $i<k$, $u_i<T-\alpha.$ As before let
  $v=\bot^{k-1}\top.$ Clearly $\det(\cA)(u)=v.$ Now consider any
  $w\in \cU- (\det(\cA))^{-1}(v)$, i.e., $\cA(w)\neq v.$ It has to be
  the case that, either for some $i<k$, $w_i\geq T$, or $w_k<T.$ In
  the former case $w_i-u_i>\alpha$ and in the later case,
  $u_k-w_k>\alpha.$ Thus, $d(u,w)>\alpha$ for every $w\in \cU$ such
  that $\det(\cA)(w)\neq v.$ Hence
  $d(u,\cU- (\det(\cA))^{-1}(v))>\alpha.$ Since $\cA$ is
  $(\alpha,\beta,0)$-accurate, it outputs $v$ with probability
  $\geq 1-\beta.$ Hence $\cA$ is $(\alpha,\beta)$-correct.
\end{proof}

Using Lemma~\ref{lem:correct-acc} and Lemma~\ref{lem:AT}, we can
conclude that {\athres} is $(\alpha,\beta,0)$-accurate for
$\beta= 2m\euler^{-\frac{\alpha\epsilon}{8}}$ and for all
$\alpha\geq 0.$

\begin{figure*}[!htb]
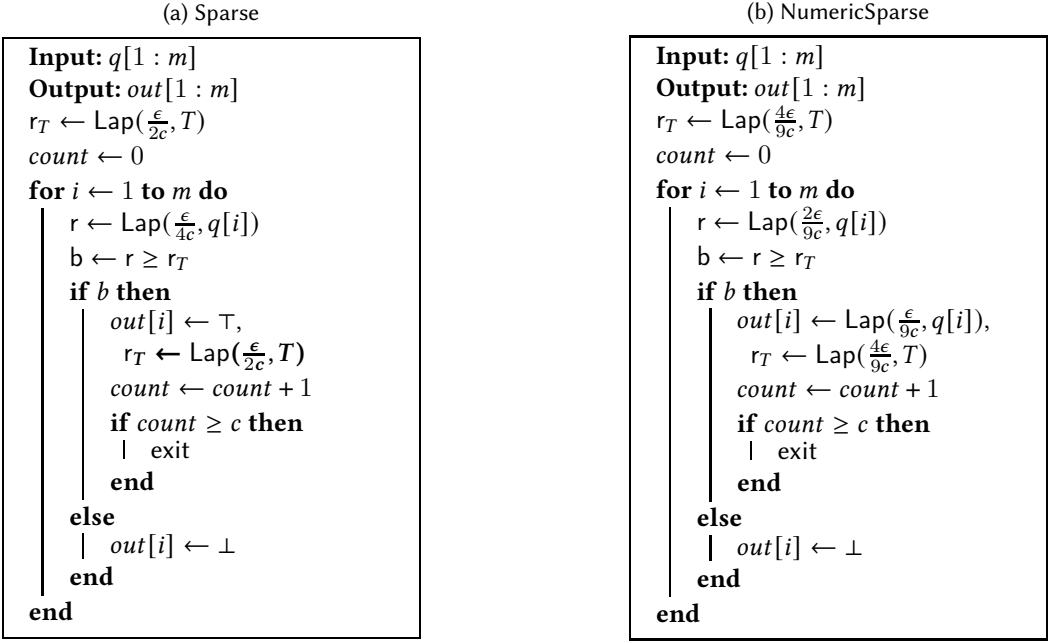

	\RestyleAlgo{boxed}
	\begin{subfigure}{0.4\textwidth}
		\caption{{\svttwo}}\label{algo-svt2}
		\removelatexerror
		\begin{algorithm}[H]
			\DontPrintSemicolon
			\KwIn{$q[1:m]$}
			\KwOut{$out[1:m]$}
			$\rv_T \gets \Lap{ \frac {\epsilon} {2c 
				} , T}$\;
			$count \gets 0$\;
			\For{$i\gets 1$ \KwTo $m$}
			{
				$\rv\gets \Lap{\frac \epsilon {4c
					} , q[i]}$\;
				$\bv\gets \rv \geq \rv_T$\;
				\uIf{$b$}{
					$out[i] \gets \top$, {\bf $\rv_T\gets \Lap{\frac \epsilon {2c
						} , T}$}\;
					$count \gets count + 1$\;
					\If{$count\geq c$} {\exit}
					}
				\Else{
					$out[i] \gets \bot$}
				}
		\end{algorithm}
	\end{subfigure}
	\hfill%
	\begin{subfigure}{0.4\textwidth}
		\caption{{\nsvt}}
		\label{algo-numeric}
		\removelatexerror
		\begin{algorithm}[H]
			\DontPrintSemicolon
			\KwIn{$q[1:m]$}
			\KwOut{$out[1:m]$}
			$\rv_T \gets \Lap{ \frac {4\epsilon} {9c 
				} , T}$\;
			$count \gets 0$\;
			\For{$i\gets 1$ \KwTo $m$}
			{
				$\rv\gets \Lap{\frac {2\epsilon} {9c 
					} , q[i]}$\;
				$\bv\gets \rv \geq \rv_T$\;
				\uIf{$b$}{
					$out[i] \gets  \Lap{\frac {\epsilon} {9c 
					} , q[i]}$, $\rv_T\gets \Lap{\frac {4\epsilon} {9c 
						} , T}$\;
					$count \gets count + 1$\;
					\If{$count\geq c$} {\exit}
					}
				\Else{
					$out[i] \gets \bot$}
			}
		\end{algorithm}
	\end{subfigure}
	\caption{\footnotesize{Algorithms {\svttwo} and {\nsvt}. {\svtone} (discussed in Section~\ref{sec:scaling}) is a variant of {\svttwo} where ${\rv_T}$ is not re-sampled in the for-loop (shown in bold here).}}
	\label{fig:two-algos}
\end{figure*}

\subsection{{\svttwo}}
\label{sec:sparse}

The {\svttwo} algorithm is a generalization of {\athres}. As in
{\athres}, we get a sequence of queries $(q_1,\ldots, q_m)$ and a
threshold $T$, and we output $\bot$ whenever the query is below $T$,
and $\top$ when it is above $T$. In {\athres} the algorithm stops when
either the first $\top$ is output or the entire sequence of queries is
processed without outputting $\top$. Now, we want to terminate when
either $c$ $\top$s are output, or the entire sequence of queries is
processed without $c$ $\top$s being output. We will call this the
deterministic function $\det(\mbox{\svttwo})$, and {\svttwo} is the
randomized version of it that preserves privacy. The algorithm
{\svttwo} is shown in Figure~\ref{algo-svt2}. The set of inputs $\cU$,
outputs $\cV$, distance metrics $d$ and $d'_u$ are the same as for
{\athres} (Section~\ref{sec:at}).

Suppose the input sequence of queries $(q_1,...,q_m)$ satisfies the
following property: for all $j$, $1\leq j\leq m$, either
$q_j < T-\alpha$ or $q_j\geq T+\alpha$, and furthermore, for at most
$c$ values of $j$, $q_j\geq T+\alpha.$ A sequence
$v=(v_1,...,v_k)\in \set{\bot,\top}^*$ is a valid output sequence for
the above input sequence of queries if $k\leq m$ and the following
conditions hold: (i) $\forall j\leq k$, $v_j=\bot$ if
$q_j(D)<T-\alpha$, otherwise $v_j=\top$; (ii) if $k<m$ then $v_k=\top$
and there are $c$ occurrences of $\top$ in $v$; (iii) if $k=m$ and
$v_k=\bot$ then there are fewer than $c$ occurrences of $\top$ in $v$.
The following accuracy claim for {\svttwo} can be easily shown by
using Theorem 3.26 of \citet{DR14}. Let $\alpha,\beta\in\Reals$ be such
that $\alpha\geq 0, \beta\in [0,1]$ and
$\alpha = \frac{8c}{\epsilon}(\ln m + \ln (\frac{2c}{\beta})).$ On an
input sequence of queries satisfying the above specified property,
with probability $\geq (1-\beta)$, {\svttwo} terminates after
outputting a valid output sequence.  Now, for $\alpha,\beta$ as given
above, using the same reasoning as in the case of the algorithm
{\athres}, it is easily shown that {\svttwo} is
$(\alpha,\beta,0)$-accurate for
$\beta\:= 2mc \euler^{-\frac{\alpha\epsilon}{8c}}$ and for all
$\alpha\geq 0.$

\subsection{{\nsvt}}
\label{sec:nsvt}

Consider a problem very similar to the one that {\svttwo} tries to
solve, where we again have a sequence of queries $(q_1,\ldots, q_m)$
and threshold $T$, but now instead of outputting $\top$ when
$q_i \geq T$, we want to output $q_i$ itself. {\nsvt} solves this
problem while maintaining differential privacy. It is very similar to
{\svttwo} and is shown in Figure~\ref{algo-numeric}. The only
difference between {\svttwo} and {\nsvt} is that instead of outputting
$\top$, {\nsvt} outputs a $q_i$ with added noise. We now show how
accuracy claims about {\nsvt} are not only captured by our general
definition, but in fact can be improved.

As before, let $\cU = \Reals^m.$ The distance function $d$ on $\cU$ is
defined as follows. For $u,u'\in \cU$, where $u=(u_1,...,u_m)$ and
$u'=(u'_1,...,u'_m)$, if for all $i$ such that both $u_i,u'_i\geq T$,
it is the case that $u_i=u'_i$, then $d(u,u')\:=\norm{(u-u')}_\infty$,
otherwise $d(u,u')=\infty.$ We define
$\cV = (\set{\bot}\cup \Reals)^m.$ The distance function $d'_u$ (for
any $u$) on $\cV$ is defined as follows: for $v,v'\in \cV$, where
$v=(v_1,...,v_m)$ and $v'=(v'_1,...,v'_m)$, if for all
$j,1\leq j\leq m,$ either $v_j=v'_j=\bot$ or $v_j,v'_j\in \Reals$ then
$d'_u(v,v') = \max \{\mathbin{|}v_j-v'_j\mathbin{|}\st v_j,v'_j\in
\Reals\}$, otherwise $d'_u(v,v')=\infty.$ Let $\det(\mbox{\nsvt})$ be
the function given by the deterministic algorithm
described earlier.

The accuracy for {\nsvt} is given by Theorem 3.28 of \citet{DR14},
which can be used to show that it is $(\alpha,\beta,\alpha)$-accurate
with $\beta= 4mc\euler^{-\frac{\alpha\epsilon}{9c}}.$ The following
theorem gives a better accuracy result in which $\beta$ is specified
as a function of both $\alpha$ and $\gamma.$ 
\ifdefined\AppendixTrue
(See  Appendix~\ref{sec:acc-ex-proofs} for a proof.) 
\fi 
In the special case, when
$\gamma=\alpha$, we get a value of $\beta$ (given in
Corollary~\ref{cor:NumSparse}) that is smaller than the above value.


\begin{theorem}
	\label{thm:nsvt}
{\nsvt} is $(\alpha,\beta,\gamma)$-accurate where
$\beta = 2mc\euler^{-\frac{\alpha\epsilon}{9c}}+c\euler^{-\frac{\gamma \epsilon}{9c}}.$
\end{theorem}


\begin{corollary} \label{cor:NumSparse} For any $\alpha>0$, {\nsvt} is
  $(\alpha,\beta,\alpha)$-accurate where $\beta = (2m+1)c\euler^{-\frac{\alpha\epsilon}{9c}}$
\end{corollary}


\subsection{\textsf{SmallDB}}
The algorithm {\textsf{SmallDB}} is given in \citet{DR14}.  Using the
notation of \citet{DR14}, we let $\cX$ denote a finite set of database
records. We let $n= \mbar\cX\mbar$ and
$\cX\:=\set{\cX_i\st 1\leq i\leq n}.$ A database $x=(x_1,...,x_n)$ is
a $n$-vector of natural numbers, where $x_i$ denotes the number of
occurrences of $\cX_i$ in the database. Thus, $\Nats^n$ is the set of
possible databases. The size of the database $x$ is simply
$\norm{x}_1.$ A query $f$ is a function $f\::\cX\rightarrow [0,1].$
The query $f$ is extended to the set of databases, by defining
$f(x)\:=\sum_{1\leq i\leq n}x_if(\cX_i).$

The algorithm {\textsf{SmallDB}} takes as input a database $x$, a
finite set of queries $\cQ$ and two real parameters $\epsilon,a>0$ and
outputs a small database from a set $\cR$, which is the set of small
databases of size $\frac{\log \mathbin{|}\cQ \mathbin{|}}{\alpha^2}.$
Our parameter $a$ is the parameter $\alpha$ of \citet{DR14}. The
algorithm employs the exponential mechanism using the utility function
$u\::\Nats^n\times \cV\rightarrow\Reals$, defined by
$u(x,v)\:= - \max \set{\mathbin{|}f(x)-f(v)\mathbin{|}\st f\in \cQ}.$
As in the {\exponential} mechanism, we take $\cU=\Nats^n$ and the
distance $d$ to be as in the {\laplace} mechanism. The distance $d'_x$
is as given in the {\exponential} mechanism. The function
$\det(\mbox{\textsf{SmallDB}})$, is defined exactly as
$\det(P^{\s{Exp}}_\epsilon)$, i.e.,
$\det(\mbox{\textsf{SmallDB}})(x)=\argmax1_v u(x,v)$. For any
$x\in \cU$, let $v_x= \det(\mbox{\textsf{SmallDB}})(x).$ Using
Proposition 4.4 of \citet{DR14} and its proof, it can easily be shown
that \textsf{SmallDB} is $(0,\beta,\gamma)$-accurate where
$$\gamma = a -
\mathbin{|}u(x,v_x)\mathbin{|}  + \frac{2}{\epsilon
  \norm{x}_1}\left(\frac{\log(\mathbin{|}X\mathbin{|})\log(\mathbin{|}Q\mathbin{|})}{a^2}
+ \log \left(\frac{1}{\beta}\right)\right).$$
        \section{The Accuracy Problem and its Undecidability}
\label{sec:undecidability}

Armed with a definition of what it means for a differential privacy
algorithm to be accurate,  we are ready to define the computational
problem(s) associated with checking accuracy claims. Let us fix a
differential privacy algorithm $P_\epsilon$ and the deterministic
algorithm $\det(P_\epsilon)$ against which it will be measured. Let us
also fix the set of inputs $\cU$ and outputs $\cV$ for
$P_\epsilon$. Informally, we would like to check if $P_\epsilon$ is
$(\alpha,\beta,\gamma)$-accurate. Typically, $\alpha,\beta,\gamma$
depend on both $\epsilon$ and the input $u$. Furthermore, the program
$P_\epsilon$ may only be well defined for $\epsilon$ belonging to some
interval $I$. Therefore, we introduce some additional parameters to
define the problem of checking accuracy.

Let $I \subseteq \Reals^{\geq 0}$ be an interval with rational end-points. Let
$\Sigma\:= \Reals^{\geq 0}\times [0,1]\times \Reals^{\geq
  0}$. Consider $\reg: I \times \cU \to \powerset{\Sigma}$, where
$\powerset{\Sigma}$ denotes the powerset of $\Sigma$. Here
$\reg(\epsilon,u)$ is the set of all valid $(\alpha,\beta,\gamma)$
triples for privacy budget $\epsilon$ and input $u$. $\reg$ shall henceforth be referred to 
as the \emph{admissible region}. 



Further, let $\din:\cU\times \cU\to \Reals^{\infty}$
be a distance function on $\cU.$ Let $\dout: \cU\times \cV\times \cV\to \Reals^\infty$ be such that 
for each $u \in \cU,$ the function $d'_u:\cV\times \cV\to \Reals^\infty$ defined as $d'_u(v_1,v_2)=\dout(u,v_1,v_2)$ is 
 a distance function on $\cV$. We will call $\din$ the \emph{input distance function} and $\dout$ the \emph{output distance function}. 
 The following two problems will be of
interest.
\begin{description}
\item[{\bf Accuracy-at-an-input:}] Given input $u$, determine if $P_\epsilon$ is $(\alpha,\beta,\gamma)$-accurate for all $\epsilon \in I$ and $(\alpha,\beta,\gamma) \in \reg(\epsilon,u)$.
\item [{\bf Accuracy-at-all-inputs:}] For all inputs $u$, determine if $P_\epsilon$ is $(\alpha,\beta,\gamma)$-accurate for all $\epsilon \in I$ and $(\alpha,\beta,\gamma) \in \reg(\epsilon,u)$.
\end{description}

\paragraph{Counterexamples.}
A counterexample for the accuracy-at-an-input decision problem, with
program $P_\epsilon$ and $u \in \cU$, is a quadruple
$(\epsilon_0,\alpha,\beta,\gamma)$ such that
$(\alpha,\beta,\gamma)\in \reg(\epsilon_0,u)$ and the program
$P_{\epsilon_0}$ is not $(\alpha,\beta,\gamma)$-accurate at input $u$,
where $\reg$ is the function specifying valid sets of parameters as
described above. Along the same lines, a counterexample for the
accuracy-at-all-inputs decision problem is a quintuple
$(u,\epsilon_0,\alpha,\beta,\gamma)$ such that
$(\epsilon_0,\alpha,\beta,\gamma)$ is a counterexample for the
accuracy-at-an-input decision problem at input $u$ for $P_\epsilon.$

The following theorem shows that the two decision problems above are
undecidable for general randomized programs $P_{\epsilon}.$ 
\ifdefined\AppendixTrue
(See Appendix~\ref{sec:undec-proof} for the proof).
\fi
\begin{theorem}
	\label{thm:undec}
Both the problems Accuracy-at-an-input and Accuracy-at-all-inputs are
undecidable for the general class of randomized programs.
\end{theorem}
%

\begin{remark}
For simplicity of presentation, we will assume that the interval $I$ is always the set $\Reals^{>0}.$
\end{remark}

\section{A Decidable class of Programs}
\label{sec:decide-pgm}

In this section we will identify a class of programs, called
{\newlang}, for which we will prove decidability results in
Section~\ref{sec:decidability}. {\newlang} programs are probabilistic
while programs that are an extension of the language {\ourlang}
introduced by~\citet{BartheCJS020}, for which checking differential
privacy was shown to be decidable. Our decidability results rely
crucially on the observation that the semantics of {\newlang} programs
can be defined using finite state parametric DTMCs, whose transition
probabilities are definable in first order logic over
reals. Therefore, we begin by identifying fragments of first order
logic over reals that are relevant for this paper
(Section~\ref{sec:rexp}) before presenting the syntax
(Section~\ref{sec:dipwhile}) followed by the DTMC semantics
(Section~\ref{sec:semantics}) for {\newlang} programs.

\subsection{Theory of Reals}
\label{sec:rexp}

Our approach to deciding accuracy relies on reducing the problem to
that of checking if a first order sentence holds on the reals. The use
of distributions like Laplace and Exponential in
algorithms, ensure that the sentences constructed by our reduction
involve exponentials. Therefore, we need to consider the full first
order theory of reals with exponentials and its sub-fragments.

Recall that $\rlin = \tuple{\Reals,0,1,+,<}$ is the first order
structure of reals, with constants $0,1$, addition, and the usual
ordering $<$ on reals. The set of first order sentences that hold in
this structure, denoted $\folreal$, is sometimes called the first
order theory of linear arithmetic. The structure
$\rpol = \tuple{\Reals,0,1,+,\times,<}$ also has multiplication, and
we will denote its first order theory,  the theory of real closed
fields, as $\foreal$. The celebrated result due to
~\citet{tarski-real} is that $\folreal$ and $\foreal$ admit
quantifier elimination and are decidable. 

\begin{definition}
	\label{def:definable-func-simple}
A partial function $f : \Reals^n  \hookrightarrow \Reals^k$ is said to be \emph{definable} in $\folreal$/$\foreal$ respectively, if there are  formulas  $\psi_{f}(\bar{x})$ and  $\varphi_f(\bar{x},\bar{y})$ over the signature of $\rlin$/ $\rpol$ respectively with $n$ and $n+k$ free variables respectively ($\bar{x}$ and $\bar{y}$ are vectors of $n$ and $k$ variables respectively) such that
	\begin{enumerate}
	\item for all  $\bar{a} \in \Reals^n, \ \bar{a} \in \dom(f) $ iff $\mathfrak{R} \models \psi_f[\bar{x} \mapsto \bar{a}]$, and
	\item for all  $\bar{a} \in \Reals^n$ such that $\bar{a} \in \dom(f)$, $f(\bar{a}) = \bar{b}$  iff  $\mathfrak{R}\models \varphi_f[\bar{x} \mapsto \bar{a}, \bar{y} \mapsto \bar{b}]$ 	\end{enumerate}
	where
	   $\mathfrak{R}$ is $\rlin$/$\rpol$ respectively. 

	%
\end{definition}

Finally, the real exponential field
$\rexp = \tuple{\Reals, 0, 1, \eulerv{(\cdot)}, +, \times}$ is the
structure that additionally has the unary exponential function
$\eulerv{(\cdot)}$ which maps $x \mapsto \eulerv{x}$. A long-standing
open problem in mathematics is whether its first-order theory (denoted
here by $\fthreal$) is decidable. However, it was shown by~\citet{mw96} that $\fthreal$ is decidable provided
Schanuel's conjecture (see ~\citet{lang}) holds for the set of reals\footnote{Schanuel's conjecture for reals states that if the real numbers $r_1,\ldots,r_n$ are linearly independent over $\Rats$ (the rationals) then the transcendence degree of the  field extension 
$\Rats(r_1,\ldots, r_n, \eulerv{r_1},\ldots, \eulerv{r_n})$  is $\geq n$ (over $\Rats$).}.

Some fragments of $\fthreal$ with the exponential function are known
to be decidable. In particular, there is a fragment identified by
~\citet{mccallum2012deciding} that we exploit
in our results. Let the language $\lngreal$ be the first-order
formulas over a restricted vocabulary and syntax defined as
follows. Formulas in $\lngreal$ are over the signature of $\rexp$,
built using variables
$\set{\epsilon} \cup \set{x_i\: |\: i \in \Nats}$. In $\lngreal$, the
unary function $\eulerv{(\cdot)}$ shall only be applied to the
variable $\epsilon$ and rational multiples of $\epsilon$. Thus, terms in $\lngreal$ are 
polynomials with rational coefficients over the variables
$\set{\epsilon} \cup \set{x_i\: |\: i \in \Nats} \cup
\set{\eulerv{\epsilon}} \cup \set{\eulerv{q\epsilon}\st q\in \Rats}$. Atomic formulas in the language are of the
form $t = 0$, $t < 0$, or $0 < t$, where $t$ is a term. Quantifier
free formulas are Boolean combinations of atomic formulas. Sentences
in $\lngreal$ are formulas of the form
\[
Q\epsilon Q_1x_1\cdots Q_nx_n \psi(\epsilon,x_1,\ldots, x_n)
\]
where $\psi$ is a quantifier free formula, and $Q$, $Q_i$s are
quantifiers~\footnote{Strictly speaking, \citet{mccallum2012deciding} allow $\eulerv{(\cdot)}$ to be applied only to $\epsilon.$ However, any sentence in $\lngreal$ with  terms of the form $\eulerv{q\epsilon}$ with $q\in \Rats$ can be easily shown to be equivalent to a formula where $\eulerv{(\cdot)}$ is applied only to $\epsilon.$ See \citet{BartheCJSVtech} for examples.}. In other words, sentences are formulas in prenex form,
where all variables are quantified, and the outermost quantifier is
for the special variable $\epsilon$. The theory $\threal$ is the
collection of all sentences in $\lngreal$ that are valid in the
structure $\rexp$. The crucial property for this theory is that it is
decidable.
\begin{theorem}[\citet{mccallum2012deciding}]
	$\threal$ is decidable.
\end{theorem}
We will denote the set of  formulas of the form $Q_jx_j\cdots Q_nx_n \psi(\epsilon,x_1,\ldots, x_n)$ by $\lngreal(\epsilon,x_1\ldots,x_{j-1}).$
Finally, our tractable restrictions (and our proofs of decidability)
shall often utilize the notion of partial functions being
\emph{parametrically definable} in $\threal$; we therefore conclude
this section with the formal definition.

\begin{definition}
	\label{def:definable-func}
A partial function $f : \Reals^n \times (0,\infty) \hookrightarrow \Reals^k$ is said to be \emph{parametrically definable} in $\threal$, if there are  formulas  $\psi_{f}(\bar{x},\epsilon)$ and  $\varphi_f(\bar{x},\epsilon,\bar{y})$ over the signature of $\rexp$ with $n+1$ and $n+k+1$ free variables respectively ($\bar{x}$ and $\bar{y}$ are vectors of $n$ and $k$ variables respectively, and $\epsilon$ is a variable) such that
	\begin{enumerate}
	\item for every $\bar{a} \in \Rats^n$, the formulas $\psi_f[\bar{x} \mapsto \bar{a}]$ and  $\varphi_f[\bar{x} \mapsto \bar{a}]$ are in $\lngreal(\epsilon,\bar y)$,
	\item for all  $\bar{a} \in \Reals^n,\ b \in (0,\infty).\ (\bar{a},b) \in \dom(f) $ iff $\rexp \models \psi_f[\bar{x} \mapsto \bar{a}, \epsilon \mapsto b]$, and
	\item for all  $\bar{a} \in \Reals^n,\ b \in (0,\infty)$ such that $(\bar{a},b) \in \dom(f)$, $f(\bar{a},b) = \bar{c}$  iff $\rexp \models \varphi_f[\bar{x} \mapsto \bar{a},\epsilon \mapsto b, \bar{y} \mapsto \bar{c}].$
	\end{enumerate}
When $n = 0$, we simply say that $f$ is \emph{definable} in $\threal$.
\end{definition}

\begin{example}
Consider the {\nsvt} algorithm from Section~\ref{sec:nsvt}. Assume that the algorithm is run on an array of size $2$ with threshold $T$ set to $0.$  Assume that the array elements take the value $x_1$ and $x_2$ with $x_1<0<x_2.$ Given $x_\gamma \geq 0$, let $p(x_1,x_2,x_\gamma,\epsilon)$ denote the probability of obtaining the output $(\bot,z)$ such that $\abs{x_2-z} < x_\gamma.$ Note that $p$ can be viewed as a partial function
   $p: R^3 \times (0,\infty) \hookrightarrow R$ with domain $\set{(x_1,x_2,x_\gamma,\epsilon)\st x_1<0<x_2, x_\gamma>0, \epsilon>0}.$ Further, $p(x_1,x_2,x_\gamma,\epsilon)$ is 
   the product $p_1 p_2$ where $$p_1=1-\frac {2} {3} (\eulerv{\frac 2 9 x_1 \epsilon }+\eulerv{-\frac{2} 9 x_2 \epsilon })+\frac 1 6 (\eulerv{\frac 4 9 x_1 \epsilon}+ \eulerv{-\frac 4 9 x_2 \epsilon}) -\frac 1 {48} (\eulerv{\frac 1 9 ( 6 x_1-2 x_2) \epsilon}+ \eulerv{-\frac 1 9 ( 6 x_2-2 x_1) \epsilon}) + \frac 1 4 \eulerv{\frac{2}{9} (x_1-x_2) \epsilon}  $$ and $p_2= 1- \eulerv{-\frac 1 9 x_\gamma \epsilon}$.
   
$p$ is definable in $\fthreal$ with $\psi_p(x_1,x_2,x_\gamma,\epsilon)$ as the formula $(x_1<0) \wedge (x_2>0) \wedge (x_\gamma>0) \wedge (\epsilon>0)$ and $\varphi_p(x_1,x_2,x_\gamma,\epsilon,y)$ as the formula $y=p_1p_2.$ Observe that for rational $q_1,q_2,c$, the formulas   $\psi_p(q_1,q_2,c,\epsilon)$ and $\varphi_p(q_1,q_2,c,\epsilon,y)$ are  in $\lngreal(\epsilon,y).$ Hence $p$ is parametrically definable.
   
Observe that $p(x_1,x_2,x_\gamma,\epsilon)$ is the probability of obtaining an output from {\nsvt} that is at most $x_\gamma$ away from the output of $\det(\nsvt)$ for inputs of the form $x_1<0<x_2.$  This probability is parametrically definable. Such an observation will be true for 
   {\newlang} programs and is a crucial ingredient in our decidability results.
\end{example}

\subsection{{\newlang} Programs}
\label{sec:dipwhile}

\begin{figure}
\begin{framed}
\raggedright

Expressions ($\bv\in \cB, \dv\in \cX, \iv\in \cZ, \rv\in \cR, d\in \sdom, i \in \integers, q\in \Rats, g\in \sF_{Bool}, f\in \sF_\sdom$):
\[
\begin{array}{lcl}
B&::=& \true \mbar \false \mbar \bv \mbar not(B) \mbar B\ and\ B \mbar B \ or \ B \mbar g(\tilde{E})\\
E&::=&   d \mbar \dv  \mbar f (\tilde{E}) \\
Z&::=& \iv \mbar  i Z  \mbar E Z \mbar Z+Z \mbar Z+i \mbar Z+ E \\
R&::=&  \rv \mbar  q R \mbar E R \mbar R+R \mbar R+q \mbar R+ E\\
\end{array}
\]

Basic Program Statements ($a\in \Rats^{>0}$, $\sim\in \set{<,>,=, \leq, \geq}$, $F$ is  a scoring function and $\chose$ is a user-defined distribution):
\[
\begin{array}{lll}
s&::= & \bv\gets B \mbar \bv\gets Z\sim Z \mbar \bv\gets Z \sim E \mbar \bv\gets  R\sim R \mbar \bv\gets  R \sim E \mbar \\
	&& \dv\gets E  \mbar \dv\gets\pexp {a\epsilon, F(\tilde{\dv}), E} \mbar \dv\gets \chose(a\epsilon, \tilde{E}) \mbar \\
    && \iv\gets Z \mbar \iv\gets\DLap{a\epsilon,E} \mbar \\
    && \rv\gets R \mbar \rv\gets\Lap{a\epsilon,E} \mbar \rv\gets\Lap{a\epsilon,\rv} \mbar\\
    && \ifstatement \bv P P \mbar   \whilestatement \bv P \mbar \ext
\end{array}
\]

Program Statements ($\ell \in \mathsf{Labels}$)
\[
\begin{array}{lcl}
P&::= &  \ell:\ s \mbar  \ell:\ s\, ;\,  P\\
\end{array}
\]
\end{framed}

\caption{\footnotesize{BNF grammar for ${\newlang}$. $\sdom$ is a finite discrete
  domain. $\sF_{Bool}$, ($\sF_{\sdom}$ resp) are set of functions that
  output Boolean values ($\sdom$ respectively).  $\cB,\cX, \cZ, \cR$
  are the sets of Boolean variables, $\sdom$ variables, integer random
  variables and real random variables. $\sLabels$ is a set of program
  labels.  For a syntactic class $S$, $\tilde{S}$ denotes a sequence
  of elements from $S$. In addition, 
  assignments to real and integer variables do not
  occur with the scope of a while statement.}}
\label{fig:BNFFull}
\end{figure}

Recently, \citet{BartheCJSVtech,BartheCJS020} identified a class of
probabilistic while programs called {\ourlang}, for which the problem
of checking if a program is differentially private is
decidable. Moreover, the language is powerful enough to be able to
describe several differential privacy algorithms in the literature that
have finite inputs and outputs. In this paper, we extend the language
slightly and prove decidability and conditional decidability results
for checking accuracy (Section~\ref{sec:decidability}). Our extension
allows for programs to have real-valued inputs and outputs ({\ourlang}
programs only have finite-valued inputs) and for these input variables
to serve as means of the Laplace mechanisms used during sampling;
{\ourlang} programs could only use $\sdom$ expressions as means of
Laplace mechanisms. The resulting class of programs, that we call
{\newlang}, is described in this section.

The formal syntax of ${\newlang}$ programs is shown in
Figure~\ref{fig:BNFFull}. Program variables can have one of four
types: $Bool$ ($\set{\true,\false}$); $\sdom$, a finite domain, which
is assumed without loss of generality to be a finite subset of
integers $\set{-N_\mx, \ldots 0, 1, \ldots N_\mx}$~\footnote{The
  distinction between Booleans and finite domain types is for
  convenience rather than technical neccessity. Moreover, $\sdom$ can
  be any finite set, including a subset of rationals.}; integers
$\integers$; and reals $\Reals$. In Figure~\ref{fig:BNFFull},
Boolean/$\sdom$/integer/real program variables are denoted by
$\cB$/$\cX$/$\cZ$/$\cR$, respectively, and
Boolean/$\sdom$/integer/real expressions are given by non-terminals
$B/E/Z/R$. Boolean expressions ($B$) can be built using Boolean
variables and constants, standard Boolean operations, and by applying
functions from $\sF_{Bool}$. $\sF_{Bool}$ is assumed to be a
collection of \emph{computable} functions returning a $Bool$. We
assume that $\sF_{Bool}$ always contains a function $\s{EQ}(x_1,x_2)$
that returns $\true$ iff $x_1$ and $x_2$ are equal. $\sdom$
expressions ($E$) are similarly built from $\sdom$ variables, values
in $\sdom$, and applying functions from the set of computable
functions $\sF_\sdom$. Next, integer expressions ($Z$) are built using
multiplication and addition with integer constants and $\sdom$
expressions, and addition with other integer expressions. Finally,
real expressions ($R$) are built using multiplication and addition
with rational constants and $\sdom$ expressions, and additions with
other real-valued expressions. One important restriction to note is
that integer-valued expressions cannot be added or multiplied in
real-valued expressions.

A {\newlang} program is a triple consisting of a set of (private)
input variables, a set of (public) output variables, and a finite
sequence of labeled statements (non-terminal $P$ in
Figure~\ref{fig:BNFFull}). Private input and public output variables
can either be of type $\sdom$ or $\Reals$; this is an important change
from {\ourlang} where these variables were restricted to be of type
$\sdom$. Thus, the set of possible inputs/outputs ($\cU$/$\cV$), is
identified with the set of valuations for input/output variables. Note
that if we represent the set of relevant variables $X'$ as a sequence
$\dv_1,\dv_2,\ldots, \dv_m$, then a valuation $val$ over $X'$ can be
viewed as a sequence $val(\dv_1),val(\dv_2),\ldots,val(\dv_m)$.

Program statements are assumed to be uniquely labeled from a set of
labels $\s{Labels}$. However, we will often omit these labels, unless
they are needed to explain something. Basic program statements
(non-terminal $s$) can either be assignments, conditionals, while
loops, or $\ext$. Statements other than assignments are
self-explanatory. The syntax of assignments is designed to follow a
strict discipline. Real and integer variables can either be assigned
the value of real/integer expressions or samples drawn using the
Laplace or discrete Laplace mechanism. An important distinction to
note between programs in {\newlang} and
{\ourlang} by~\citet{BartheCJS020}, is that when sampling using Laplace,
real variables in addition to $\sdom$ expressions can be used as the
mean. $\sdom$ variables are either assigned values of $\sdom$
expressions or sampled values. Sampled values for $\sdom$ variables
can either be drawn using an exponential mechanism
($\pexp {a\epsilon, F(\tilde{\dv}), E}$) with a rational-valued,
computable scoring function $F$, or a user-defined distribution
($\chose(a\epsilon,\tilde{E})$), where the probability of picking a
value $d$ as function of $\epsilon$ according to $\chose$ is parametrically definable
in $\threal$ as a function of $\epsilon$. Moreover, we assume that there is an algorithm that on
input $a, \tilde{d}$ returns the formula defining the probability of
sampling $d\in \sdom$ from the distribution
$\chose(a\epsilon,\tilde{d})$, where $\tilde{d}$ is a sequence of
values from $\sdom$. For assignments to Boolean variables, it is worth
directing attention to the cases where a variable is assigned the
result of comparing two expressions. Notice that the syntax does not
allow comparing real and integer expressions. This is an important
restriction to get decidability. For technical convenience, we assume
that in any execution, variables appearing on the right side of an
assignment are assigned a value earlier in the execution.

In addition to the syntactic restrictions given by the BNF grammar in
Figure~\ref{fig:BNFFull}, we require that {\newlang} programs satisfy
the following restriction; this restriction is also used in defining
{\ourlang} by~\citet{BartheCJS020}.
\begin{description}
\item[Bounded Assignments] Real and integer variables are not assigned
  within the scope of a while loop. Therefore, real and integer
  variables are assigned only a \emph{bounded} number of times in any
  execution. Thus, without loss of generality, we assume that real and
  integer variables are assigned \emph{at most once}, as a program
  with multiple assignments to a real/integer variable can always be
  rewritten to an equivalent program where each assignment is to a
  fresh variable.
\end{description}

\noindent
{\bf {\ourlang} and {\newlang}.}
{\ourlang}, introduced by~\citet{BartheCJS020}, is a  rich language that can describe differential privacy mechanisms with finitely many input variables taking values over a finite domain, and output results over a finite domain.
Programs can sample from continuous and discrete versions of Laplacian distributions, user-defined distributions over $\sdom$ and exponential mechanism distributions with finite support. Any {\ourlang} program can be rewritten as a program in which variables are initially sampled from Laplacian distributions, comparisons between linear combinations of sampled values and inputs are stored in Boolean variables, followed by steps of a simple probabilistic program with Boolean and $\sdom$ variables.     
 {\ourlang} can express several differential privacy mechanisms such as the algorithms {\athres}, {\nmax}, {\svttwo} and exponential mechanism discussed in Section~\ref{sec:accuracy-ex}. Other examples expressible in {\ourlang} include private vertex cover~\cite{GLMRT10} and randomized response. It can also, for example, express versions of {\nmax} where the noise is sampled from an exponential distribution and not from a Laplacian distribution. 
In {\newlang}, we allow inputs to take real values. Further, we allow programs to output real values formed by linear combinations of input and sampled real variables. This allows us to express mechanisms such as 
the private smart sum
algorithm~\cite{CSS10} and {\nsvt} (See Section~\ref{sec:nsvt}). In {\ourlang}, we could only approximate these examples by discretizing the output values and restricting the input variables to take values in $\sdom$.  
{\newlang} does not allow using Gaussian mechanisms to sample, primarily because our decision procedures do not extend to such algorithms.

\begin{example}
	\label{ex:numspaL}
	Algorithm~\ref{fig:numspaL} shows how {\nsvt} can be encoded in our
  language with $T=0,\delta=0, N=2, c=1$; this is a specialized
  version of the pseudocode in Figure~\ref{algo-numeric}. The
  algorithm either outputs $\bot$ or a numeric value. We don't have
  variables of such a type in our language. We therefore encode each
  output as a pair: $\sdom$ variable $\s{o}^1$ and real variable
  $\s{o}^2$. If $\s{o}^1=0$ then output is $\bot$ and if $\s{o}^1 = 1$
  then the output is $\s{o}^2$. Though for-loops are not part of our
  program syntax, they can be modeled as while loops, or if bounded (as
  they are here), they can be unrolled.
\end{example}

\begin{figure*}
    \resizebox{!}{2in}{
	\begin{minipage}[T]{0.35\textwidth}
		\IncMargin{0.5em}

		\RestyleAlgo{boxed}
		\removelatexerror
		\begin{algorithm}[H]
			\DontPrintSemicolon

			\KwIn{$q_1,q_2$}
			\KwOut{$\s{o}^1_1, \s{o}^2_1, \s{o}^1_2, \s{o}^2_2$}
			\;

			\SetKwSty{textsfsf}
			\nl $T\gets 0$;\;
			\nl  $\s{o}^1_1\gets 0$;\;
			\nl  $\s{o}^1_2\gets 0$;\;
			\nl  $\rv_T \gets \Lap{ \frac{4\epsilon}{9} , T};$\;
			\nl   $\rv_1\gets \Lap{\frac{2\epsilon}{9} , q_1}$;\;
			\nl $\bv\gets \rv_1 \geq \rv_T$;\;

			\nl \uIf{$\bv$}{
				\nl    $\s{o}^1_1 \gets 1$\;
				\nl    $\s{o}^2_1 \gets \Lap{\frac{\epsilon}{9},q_1}$\;
			}
			\Else
			{
				\nl   $\rv_2 \gets \Lap{\frac {2\epsilon} {9} , q_2}$;\;
				\nl  $\bv\gets \rv_2 \geq \rv_T$;\;
				\nl  \If{$\bv$}{
					\nl    $\s{o}^1_2 \gets 1$\;
					\nl    $\s{o}^2_2 \gets \Lap{\frac{\epsilon}{9},q_2}$\;
				}
			}
			\nl    $\ext$\;
			\caption{\footnotesize{{\nsvt} with $N=2$, $c=1$, $\delta=0$, and $T=0$. The numbers at the beginning of a line indicate the label of the statement.}}
			\label{fig:numspaL}
		\end{algorithm}
		\DecMargin{0.5em}
	\end{minipage}
	}
	\qquad
	\resizebox{!}{2in}{
	\begin{minipage}[T]{0.6\textwidth}
		\setlength{\tabcolsep}{1pt}
		\begin{center}
			\begin{tikzpicture}[
			every text node part/.style={align=center},
			state/.style={draw,rounded corners,minimum width=4.1cm},
			trans/.style={above,pos=0.65}]
			\footnotesize
			\node(d1) at (0,8.2) {$\vdots$};
			\node[state](s5) at (0,7) {%
				\begin{tabular}{ll}
				{\bf 10}: & $q_1:u$, $q_2:v$, $T:0$\\
				& $\s{o}^1_1: 0$, $\s{o}^1_2: 0$, $b: \bot$\\
				& $r_T: (\frac{4}{9},0)$ $r_1: (\frac{2}{9},u)$\\
				& $r_1 < r_T$
				\end{tabular}};
			\node[state](s6) at (0,4) {%
				\begin{tabular}{ll}
				{\bf 11}: & $q_1:u$, $q_2:v$, $T:0$\\
				& $\s{o}^1_1: 0$, $\s{o}^1_2: 0$, $b: \bot$\\
				& $r_T: (\frac{4}{9},0)$ $r_1: (\frac{2}{9},u)$
				$r_2: (\frac{2}{9},v)$\\
				& $r_1 < r_T$
				\end{tabular}};
			\node[state](s71) at (-2.5,1) {%
				\begin{tabular}{ll}
				{\bf 12}: & $q_1:u$, $q_2:v$, $T:0$\\
				& $\s{o}^1_1: 0$, $\s{o}^1_2: 0$, $b:\top$\\
				& $r_T: (\frac{4}{9},0)$ $r_1: (\frac{2}{9},u)$
				$r_2: (\frac{1}{4},v)$\\
				& $r_1 < r_T$, $r_2 \geq r_T$
				\end{tabular}};
			\node[state](s72) at (2.5,1) {%
				\begin{tabular}{ll}
				{\bf 12}: & $q_1:u$, $q_2:v$, $T:0$\\
				& $\s{o}^1_1: 0$, $\s{o}^1_2: 0$, $b:\bot$\\
				& $r_T: (\frac{4}{9},0)$ $r_1: (\frac{2}{9},u)$
				$r_2: (\frac{2}{9},v)$\\
				& $r_1 < r_T$, $r_2 < r_T$
				\end{tabular}};
			\node(d21) at (-2.5,0) {$\vdots$};
			\node(d22) at (2.5,0) {$\vdots$};
			\draw[->] (s5) -- node[left]{1} (s6);
			\draw[->] (s6) -- node[trans]{$p$} (s71);
			\draw[->] (s6) -- node[trans]{$q$} (s72);
			\end{tikzpicture}
		\end{center}
		\caption{\footnotesize{Partial DTMC semantics of
        Algorithm~\ref{fig:numspaL} showing the steps when lines 10
        and 11 are executed. $q_1$ and $q_2$ are assumed to have
        values $u$ and $v$, respectively. Only values of assigned
        program variables are shown. Third line in state shows
        parameters for the real values that were sampled. Last line
        shows the accumulated set of Boolean conditions that hold on
        the path.}}
		\label{fig:numspa-dtmc}
	\end{minipage}
	}
\end{figure*}

\subsection{Semantics}
\label{sec:semantics}

A natural semantics for {\newlang} programs can be given using Markov
kernels. Given a fixed $\epsilon > 0$, the states in such a semantics
for program $P_\epsilon$ will be of the form
$(\ell, h_{Bool}, h_{\sdom},$ $h_{\integers}, h_{\Reals})$, where
$\ell$ is the label of the statement of $P_\epsilon$ to be executed
next, the functions $h_{Bool}$, $h_{\sdom}$, $h_{\integers}$, and
$h_{\Reals}$ assign values to the Boolean, $\sdom$, real, and integer
variables of the program $P_\epsilon$. There is a natural
$\sigma$-algebra that can be defined on such states, and the semantics
defines a Markov kernel over this algebra. Such a semantics for
{\newlang} would be similar to the one for {\ourlang} given
by~\citet{BartheCJSVtech}, and is skipped here. Throughout the paper, we shall also assume that 
{\newlang} programs terminate with probability $1$ on all inputs.  

Our decidability results rely crucially on the observation that the
semantics of {\newlang} can be defined using a \emph{finite-state}
(parametrized) DTMC. This semantics, though not natural, can be shown
to be equivalent to the Markov kernel semantics. The proof of
equivalence is similar to the one given by~\citet{BartheCJSVtech}. We
spend the rest of this section highlighting the main aspects of the
DTMC semantics that help us underscore the ideas behind our decision
procedure. We begin by recalling the definition of a finite-state
parametrized DTMC.
\begin{definition}
	\label{def:pDTMC}
	A \emph{parametrized DTMC} over $(n+1)$ parameters
  $(\overline{x},\epsilon)$ is a pair $\cD = (\states,\ptransf)$,
  where $\states$ is a finite set of states, and
  $\ptransf: \states \times \states \to (\Reals^n \times (0,\infty)
  \to [0,1])$ is the \emph{probabilistic transition function}. For any
  pair of states $\stat,\stat'$, $\ptransf(\stat,\stat')$ will be
  called the \emph{probability of transitioning} from $\stat$ to
  $\stat'$, and is a function that, given $\overline{a} \in \Reals^n$
  and $b \in (0,\infty)$, returns a real number between $0$ and $1$,
  such that for any state $\stat$,
  $\sum_{\stat' \in \states}\ptransf(\overline{a},b)(\stat,\stat') =
  1$.
\end{definition}
The connection between programs in {\newlang} and parametrized DTMCs
is captured by the following result that is exploited in our
decidability results.
\begin{theorem}
\label{thm:semantics}
Let $P_\epsilon$ be an arbitrary {\newlang} program whose real-valued
input variables are $\overline{x}$. There is a finite state
parametrized DTMC $\sem{P_\epsilon}$ over parameters
$(\overline{x},\epsilon)$ (with transition function $\ptransf$) that
is equivalent to the Markov kernel semantics of $P_\epsilon$. Further,
the DTMC $\sem{P_\epsilon}$ is effectively constructible, and for any
pair of states $\stat,\stat'$, the partial function
$\ptransf(\stat,\stat')$ is parametrically definable in $\threal$.
\end{theorem}
\begin{proof}[Proof Sketch]
	The formal construction of the parametrized DTMC $\sem{P_\epsilon}$
  is very similar the one outlined in~\citet{BartheCJSVtech} for (the
  restricted) {\ourlang} programs. Here, we just sketch the main
  ideas. It is useful to observe that defining a \emph{finite-state}
  semantics for {\newlang} programs is not obvious, since these
  programs have real and integer valued variables. The key to
  obtaining such a finite state semantics is to not track the values
  of real and integer variables explicitly, but rather implicitly
  through the relationships they have amongst each other.

	Informally, a state in $\sem{P_\epsilon}$ keeps track of a program
  statement to be executed (in terms of its label), and the values
  stored in each of the Boolean and $\sdom$ variables. However, the
  values of real and integer variables will not be explicitly stored
  in the state. Recall that real and integer variables are assigned a
  value only once in a {\newlang} program. Therefore, states of
  $\sem{P_\epsilon}$ store the (symbolic) expression on the right side
  of an assignment for each real/integer variable, instead of the
  actual value; when the value is sampled, the symbolic parameters of
  the distribution are stored. In addition to symbolic values for real
  and integer variables, a state of $\sem{P_\epsilon}$ also tracks the
  relative order among the values of real and integer variables. Thus,
  $\sem{P_\epsilon}$ has only finitely many states. A state of
  $\sem{P_\epsilon}$ is an abstraction of all ``concrete states''
  whose assignments to Boolean and $\sdom$ variables match, and whose
  assignment to real and integer variables satisfy the constraints
  imposed by the symbolic expressions and the relative order
  maintained in the $\sem{P_\epsilon}$ state.

	State updates in $\sem{P_\epsilon}$ are as follows. Assignments to
  $\sdom$ variables are as expected --- a new value is calculated and
  stored in the state for a deterministic assignment, or a value is
  sampled probabilistically and stored in the state for a randomized
  assignment. Assignments to real and integers variables are
  \emph{always} deterministic --- the state is updated with the
  appropriate symbolic values that appear in the deterministic or
  probabilistic assignment. It is important to note that sampling a
  value using a Laplace mechanism is a deterministic step in the DTMC
  semantics. Assignments to Boolean variables, where the right hand
  side is a Boolean expression, is as expected; the right hand side
  expression is evaluated and the state is updated with the new
  value. Assignments to Boolean variables by comparing two real or
  integer expressions is handled in a special way. These are
  \emph{probabilistic transitions}. Consider an assignment
  $\bv\gets R_1\sim R_2$ for example. The result of executing this
  statement from state $\stat$ will move to a state where
  $R_1 \sim R_2$ is added to the set of ordering constraints, with
  probability equal to the probability that $R_1 \sim R_2$ holds
  conditioned on the ordering constraints in $\stat$ holding, subject
  to the variables being sampled according to the parameters stored in
  $\stat$. With the remaining probability, $\sem{P_\epsilon}$ will
  move to a state where $\neg(R_1 \sim R_2)$ is added to the ordering
  constraints. Finally, branching and while statements are
  deterministic steps with the next state being determined by the
  value stored for the Boolean variable in the condition.

	Notice here that since input variables and the privacy parameter
  $\epsilon$ can appear as parameters of the Laplace/discrete Laplace
  mechanism used to sample a value of a real/integer variable, the
  transition probabilities of $\sem{P_\epsilon}$ depend on these
  parameters. That these transition probabilities are parametrically 
  definable in $\threal$ can be established 
  along the same lines as the proof that the transition probabilities are 
  definable in $\threal$ for the DTMC semantics of $\ourlang.$    
\end{proof}

\begin{example}
	\label{ex:numspa-dtmc}
	The parametrized DTMC semantics of Algorithm~\ref{fig:numspaL} is
  partially shown in Figure~\ref{fig:numspa-dtmc}. We show only the
  transitions corresponding to executing lines 10 and 11 of the
  algorithm, when $q_1 = u$ and $q_2 =v$ initially; here
  $u,v \in \{\bot,\top\}$. The multiple lines in a given state give
  the different components of the state. The first two lines give the
  assignment to $Bool$ and $\sdom$ variables, the third line gives
  values to the integer/real variables, and the last line has the
  Boolean conditions that hold along a path. Since 10 and 11 are in
  the else-branch, the condition $r_1 < r_T$ holds. Notice that values
  to real variables are not explicit values, but rather the parameters
  used when they were sampled. Finally, observe that probabilistic
  branching takes place when line 11 is executed, where the value of
  $b$ is taken to be the result of comparing $r_2$ and $r_T$. The
  numbers $p$ and $q$ correspond to the probability that the
  conditions in a branch hold, given the parameters used to sample the
  real variables and \emph{conditioned} on the event that $r_1 < r_T$.
\end{example}


%
%
%
%
%
%
%
%

\newcommand{\agsimple}{simple}

\section{Deciding Accuracy for DipWhile+ Programs}
\label{sec:decidability}

We shall now show that the problem of checking the accuracy of a {\newlang} program is decidable, assuming Schanuel's conjecture. Further, we shall identify special instances under which the problem of checking the accuracy of a {\newlang} program is decidable without assuming Schanuel's conjecture. Our results are summarized in Table~\ref{tab:test}.

\begin{remark}
For the rest of the section, we shall say that inputs/outputs to the  {\newlang} program $P_\epsilon$  are rational if all the real variables in the input/output respectively take rational values.
We shall also say that $P_\epsilon$
has finite inputs if all of its input variables are $\sdom$-variables, and  that $P_\epsilon$ has finite outputs if all of its output variables are  $\sdom$-variables.
\end{remark}

\begin{table*}[!t]
 \caption{\footnotesize{Summary of our decidability results.
   The column, Schanuel, indicates whether the result is conditional on
    Schanuel's conjecture. The column, Problem, indicates if the decision problem is Accuracy-at-all-inputs or Accuracy-at-an-input. The column, Infinite Inputs, indicates if the result allows real variables as inputs. Note that this column is relevant only for the Accuracy-at-all-inputs decision problem. The column, Infinite Outputs, indicates if the result allows real variables as outputs.  The columns, $\det(P_\epsilon),\din$ and $\dout,$ indicate the definability assumptions needed for deterministic function $\det(P_\epsilon),$ input distance function $\din$ and output distance function $\dout.$ The column, Region $\reg$ indicates the assumptions needed on admissible region.}
}  \label{tab:test}
  \centering
  \footnotesize
  \begin{tabular}[t]{| c | c | c | c | c | c | c | c | c |}
    \hline
    Result & Problem & Schanuel & \shortstack{Infinite\\ Inputs} & \shortstack{Infinite \\ Outputs} & $\det(P_\epsilon)$ & $\din$ & $\dout$ & Region $\reg$ \\ \hline

    Thm~\ref{thm:condAcc} & all-inputs & \cmark & \cmark & \cmark & $\foreal$ & $\foreal$ & $\folreal$ & param. def. in $\threal$ \\ \hline

    Cor~\ref{cor:finite} & all-inputs & - & \xmark & \xmark & $\folreal$ & $\folreal$ & $\folreal$ &  $(\alpha,\gamma)$-monotonic \\ \hline

    Cor~\ref{cor:accpoint} & an-input & \cmark & - & \cmark & $\foreal$ & $\foreal$ & $\folreal$ & param. def. in $\threal$ \\ \hline

    Thm~\ref{thm:infpoint} & an-input & - & - & \cmark & $\folreal$ &
                                                                  $\foreal$ & $\folreal$ & \shortstack{simple, fixed $\alpha$, \\ fixed $\gamma$} \\ \hline

    Thm~\ref{thm:fixedgamma} & an-input & - & - & \cmark & $\folreal$ & $\folreal$ & $\folreal$ & \shortstack{limit-def., fixed $\gamma$} \\ \hline

    Thm~\ref{thm:finiteoutputs} & an-input & - & - & \xmark & $\folreal$ & $\folreal$ & $\folreal$ & $(\alpha,\gamma)$-monotonic \\ \hline
  \end{tabular}
 
\end{table*}

\subsection{Definability assumptions}

Let $P_\epsilon$ be a {\newlang} program with $\ell$ input $\sdom$-variables, $k$ input real variables, $m$ output $\sdom$-variables, and $n$ output real variables.
Observe that $\sdom$ is a subset of integers. Hence, the input set $\cU=\sdom^\ell\times \Reals^k$ can be viewed as a subset of $\Reals^{\ell+k}$, and
the output set $\cV$ as a subset of  $\Reals^{m+n}.$ Thus, $\det(P_\epsilon)$ can be viewed as a partial function from $\Reals^{\ell+k}$  to $\Reals^{m+n}$.


Also, observe that $\Reals^\infty$ can be seen as a subset of $\Reals\times \Reals$ by
identifying $r\in \Reals$ with  $(0,r)$ and $\infty$ with $(1,0).$
Thus, the input distance function, $\din$,  can be viewed as a partial function from $\Reals^{\ell+k} \times \Reals^{\ell+k}$ to $\Reals^2$ and the output distance function, $\dout$, as a partial function from $\Reals^{\ell+k}\times\Reals^{m+n} \times \Reals^{m+n}$ to $\Reals^2.$ Our results shall require that these functions be definable in sub-theories of real arithmetic.

Finally, recall that the admissible region $\reg$ is a function that given a privacy budget $\epsilon$ and input $\bar u$ gives the set of the set of all valid  $(\alpha,\beta,\gamma)$ for that $\epsilon$ and input $\bar u.$ Observe that $\reg$ can also be viewed as a function that takes $\alpha,\beta,\gamma,\bar u$ and $\epsilon$ as input and returns $1$ if the triple  $(\alpha,\beta,\gamma)$ is in the set $\reg(\epsilon,\bar u)$ and $0$ otherwise.

\begin{definition} Let $P_\epsilon$ be a {\newlang} program implementing the deterministic function $\det (P_\epsilon).$ Let $\din$ be a distance function on the set of inputs of $P_\epsilon$ and  $\dout$ be the input-indexed distance function on the set of outputs of $P_\epsilon$. Let $P_\epsilon$ have $\ell$ input $\sdom$-variables, $k$ input real variables, $m$ output $\sdom$-variables, and $n$ real output variables. Let $\reg$ denote the admissible region.
\begin{itemize}
\item $\det (P_\epsilon)$ is said to be definable in $\foreal$ ($\folreal$ respectively) if it is definable in $\foreal$ ($\folreal$ respectively) when viewed as a partial function from $\Reals^{\ell+k}$  to $\Reals^{m+n}$.
\item $\din$ is said to be definable in $\foreal$ ($\folreal$ respectively) if it is definable in   $\foreal$ ($\folreal$ respectively) when viewed as a partial function from $\Reals^{\ell+k} \times \Reals^{\ell+k}$ to $\Reals^2.$
\item $\dout$  is said to be definable in $\foreal$ ($\folreal$ respectively) if it is definable in   $\foreal$ ($\folreal$ respectively) when viewed as a partial function from $\Reals^{\ell+k}\times\Reals^{m+n} \times \Reals^{m+n}$ to $\Reals^2.$
\item The admissible region $\reg$ is said to be parametrically definable in $\threal$ if the partial function $h: \Reals^{3+\ell+k}\times (0,\infty) \rightarrow \Reals$ defined as
\[
   h(x,y,z,\bar u,\epsilon) =\begin{cases} 1& \mbox{ if } \epsilon>0 \mbox{ and } (x,y,z) \in \regf \epsilon {\bar{u}} \\ 0 & \mbox { otherwise} \end{cases}
\]
is parametrically definable in $\threal.$
\end{itemize}
\end{definition}

Now, we could have chosen to write $\det(P_\epsilon)$ in {\newlang} by considering programs that do not contain any probabilistic assignments.  Please note that a deterministic program written in {\newlang} can be defined in $\folreal.$ Intuitively, this is because we do not allow assignments to integer and real random variables inside loops of {\newlang} programs. Hence the loops can be ``unrolled''. This means that
 there are only finitely many possible executions of a deterministic {\newlang} program, and these executions have finite length. The Boolean checks in the program determine which execution occurs on an input, and these checks can be encoded as formulas in linear arithmetic.

\subsection{Decidability  assuming Schanuel's conjecture}
We start by establishing the following Lemma, which says that if a {\newlang} program has only finite outputs (i.e., only $\sdom$-outputs) then the probability of obtaining an output $\bar v$ is parametrically definable in $\threal.$ The proof of this fact essentially mirrors the proof of the fact that this probability is definable in $\threal$ (without parameters) for the (restricted) {\ourlang} programs established in~\citet{BartheCJSVtech}. It based on the observation that it suffices to compute the probability of reaching certain states (labeled exit states) of the DTMC semantics, which have $\bar v$ as the valuation over output variables. The reachability probabilities can be computed as a solution to a linear program (with transition probabilities as the coefficients).

\begin{lemma}
\label{lem:ProbDef}
Let $P_\epsilon$ be a {\newlang} program with finite outputs. Let $P_\epsilon$ have $\ell$ input $\sdom$-variables, $k$ input real variables and $m$ output variables.
Given,  $\bar v\in \sdom^m$, let $\mathsf{Pr}_{\bar v, P_\epsilon}: \Reals^{\ell+k+1} \hookrightarrow \Reals $ be the partial function whose domain is $\sdom^\ell\times \Reals^{k+1},$ and which
 maps $(\bar r,\epsilon)$ to the probability that $P_\epsilon$ outputs $\bar v$ on input $\bar r.$ For each $\bar v$, the function $\mathsf{Pr}_{\bar v,P_\epsilon}$ is parametrically definable in $\threal.$

\end{lemma}

The following result gives sufficient conditions under which the decision  problem Accuracy-at-all-inputs is decidable for {\newlang} programs. The conditions state that the deterministic
program and the input distance function be definable using first-order theory of real arithmetic. The output distance distance function is required to be definable in the first-order theory of linear arithmetic. Intuitively, this additional constraint is needed as it implies that we only need to compute probabilities that the outputs reside in a region defined by linear equalities and linear inequalities.
\begin{theorem}
\label{thm:condAcc}
Assuming Schanuel's conjecture, the problem Accuracy-at-all-inputs is decidable for {\newlang} programs $P_\epsilon$ when (a) $\det(P_\epsilon)$ is definable in $\foreal,$ (b) $\din$ is definable in $\foreal$, (c) $\dout$ is definable in $\folreal,$ and (d) $\reg$ is parametrically definable in   $\threal.$

\end{theorem}


The  problem  Accuracy-at-an-input is also decidable under the same constraints as given by Theorem~\ref{thm:condAcc}. This is established as a corollary to the proof of Theorem~\ref{thm:condAcc}. 
\ifdefined\AppendixTrue
(See Appendix~\ref{app:condAcc} for the proof).
\fi
\begin{corollary}
\label{cor:accpoint}
Assuming Schanuel's conjecture is true for reals,  the problem Accuracy-at-an-input is  decidable for {\newlang} programs $P_\epsilon$ and rational inputs $\bar u$ when (a) $\det(P_\epsilon)$ is definable in $\foreal,$ (b) $\din$ is definable in $\foreal,$ (c) $\dout$ is definable in $\folreal,$ and (d) $\reg$ is parametrically definable in   $\threal.$
\end{corollary}

\subsection{Unconditional Decidability Results}
We shall now give sufficient conditions under which the problems Accuracy-at-an-input and  Accuracy-at-all-inputs will be decidable \emph{unconditionally}, i.e., without assuming Schanuel's conjecture. For these results, we will have to restrict the admissible region. All examples discussed in Section~\ref{sec:accuracy-ex} have regions that satisfy these restrictions under reasonable assumptions.
We start by defining one restriction on regions that will be needed by all our unconditional decidability results. Intuitively, this restriction says that $\alpha,\gamma$  are independent of the privacy budget, while $\beta$ is a function of $\alpha,\gamma,\epsilon$  and the input. In addition, we require that $\beta$ in the region is anti-monotonic in $\alpha$ and $\gamma$ (condition 5 below).

 \begin{definition}
\label{def:simple}
The admissible region $\reg$ is {\agsimple} if there is a partial function $\Iag: \Reals^p\hookrightarrow \powerset{\Reals^{\geq 0}\times\Reals^{\geq 0}}$ and a partial function
$f_\beta:\Reals^{2+p}\times (0,\infty) \hookrightarrow [0,1]$ such that
\begin{enumerate}
\item $domain (\Iag) = \cU$ where $\cU\subseteq \Reals^p,$
\item $domain(f_\beta)=\set{(a,c,\bar u,\epsilon)\st \bar u\in \cU, (a,c)\in \Iag(\bar u), \epsilon>0) },$
\item $f_\beta$ is parametrically definable in $\threal$,
\item $(a,b,c) \in \regf \epsilon {\bar u}$ iff $(a,c)\in \Iag(\bar u), \bar u \in \cU$ and $f_\beta(a,c,\bar u,\epsilon)=b$, and
\item for all $\epsilon,\bar{u},a_1,a_2,c_1,c_2$ with $(a_i,c_i) \in \Iag(\bar u)$ for $i\in \set{1,2}$, and $a_1\leq a_2, c_1\leq c_2$, $f_\beta(a_2,c_2,\bar u,\epsilon) \leq f_\beta(a_1,c_1,\bar u,\epsilon)$.
\end{enumerate}
\end{definition}

\begin{example}
\label{example:simple}
Let $P_\epsilon$ be a program with $\cU$ as the set of inputs.
Let $\eta_1$ be the region defined as follows. For each $\epsilon>0$ and input $\bar u \in \cU$,
$
\eta_1(\bar u, \epsilon)= \set{(a,b,0) \st a\geq 0, b=\euler^{-\frac {a  \epsilon} 2}}.
$
$\eta_1$ is a simple region with 
$$\Iag(\bar u)=\begin{cases}
               \set{(a,0)\st a \geq 0} & \mbox{if } \bar u \in \cU \\
               \mbox{undefined} &\mbox{otherwise}
               \end{cases}$$ and 
$$f_\beta(a,c,\bar u,\epsilon) =\begin{cases}
                                 \euler^{-\frac {a  \epsilon} 2} & \mbox{if } a\in \Reals^{\geq 0}, c=0, \bar u \in \cU, \epsilon >0\\
                                 \mbox{undefined} &\mbox{otherwise}.
                                \end{cases}$$
Notice that, since $\euler^{-\frac{a_2 \epsilon} 2} \leq \euler^{-\frac{a_1 \epsilon} 2}$ if $a_1 \leq a_2$, the region $\eta_1$ satisfies condition 5 of Definition~\ref{def:simple}.

 On the other hand, the region $\eta_2$ defined as
 $
\eta_2(\bar u, \epsilon)= \set{(a,b,0) \st a\geq \epsilon, b=\euler^{-\frac {a  \epsilon} 2}}
$
is not a simple region because $\alpha$ depends on $\epsilon.$ 
\end{example}

\begin{remark}
For the rest of this section, we assume that $\reg$ is represented by the pair $(\Iag,f_\beta).$ For inputs to decision problems, $f_\beta$ will be represented by the formulas
$(\psi_{\beta},\phi_{\beta})$ defining it. $\Iag$ will  usually represented by a first-order formula $\Tag(x_\alpha,x_\gamma,\bar x)$ such that for all $a,c,\bar u$, $(a,c)\in\Iag(\bar u)$ iff $\Tag(a,c,\bar u)$ is true. 
\end{remark}

\subsubsection*{Program with infinite outputs}

We start by showing that the problem of checking accuracy for {\newlang} programs at a rational input $\bar u$ is decidable when we fix $\alpha,\gamma$ to be some rational numbers. For this result, we shall require that the deterministic function $\det (P_\epsilon(\bar u))$ be definable in $\folreal.$ This implies that the output  of the function at $\bar u$ must be rational. The proof essentially requires that the program $P^{\s{new}}_\epsilon$ constructed in the proof  of Theorem~\ref{thm:condAcc} be executed on $\bar u,$ $\det (P_\epsilon(\bar u))$ and $\gamma.$ The assumption that $\det (P_\epsilon(\bar u))$ is definable in $\folreal$ will ensure that the inputs to $P^{\s{new}}_\epsilon$ are rational numbers. 
 Thus, the probability of $P_\epsilon$ generating an output on input $\bar u$ that is at most $\gamma$ away, can then be defined in $\threal.$
This observation, together with the  parametric definability of $f_\beta$ allows us to show that the sentence constructed in the proof of Corollary~\ref{cor:accpoint} that checks accuracy at $\bar u$ is a sentence in $\lngreal.$  In the decision procedure, we need to provide only the fixed values of $\alpha$ and $\gamma $ as a description of
$\Iag.$ 
\ifdefined\AppendixTrue
The formal proof can be found in Appendix~\ref{app:infpoint}.
\fi
\begin{theorem}
\label{thm:infpoint}
The problem Accuracy-at-an-input is  decidable for {\newlang} programs $P_\epsilon$ and rational inputs $\bar u$ when (a) $\det(P_\epsilon)$ is definable in $\folreal,$ (b) $\din$ is definable in $\foreal,$ (c) $\dout$ is definable in $\folreal,$ (d) $\reg=(\Iag, f_\beta)$ such that $\reg$ is {\agsimple} and $\Iag ({\bar u}) = \set{(a,c)}$   for some rational numbers $a,c.$

\end{theorem}

One natural question is if the above result can be extended to checking accuracy for varying $\alpha,\gamma.$ We shall show that, with additional restrictions, we can  establish decidability of Accuracy-at-an-input when only $\gamma$ is fixed. Intuitively, this result will exploit the fact that for a given input $\bar u$, the interesting $\alpha$ to consider is the distance to disagreement for $\bar u$ (as $\beta$ decreases with increasing $\alpha$). We can then proceed as in the proof of Theorem~\ref{thm:infpoint}. 
This idea mostly works except that one has to ensure that the distance to disagreement is a rational number to apply Theorem~\ref{thm:infpoint}; the function $f_\beta$ does not jump at the point of disagreement; and that we can compute $f_\beta$ at $\infty$
(for the case when the distance to disagreement is $\infty$). The following definition captures the latter two restrictions.

\begin{definition}

The {\agsimple} region $\reg=(\Iag,f_\beta)$ is said to be limit-definable  if

\begin{enumerate}
\item There is a linear arithmetic formula $\Tag(x_\alpha,x_\gamma,\bar x)$ such that for all $a,c,\bar u$, $(a,c)\in\Iag(\bar u)$ iff $\Tag(a,c,\bar u)$ is true.
\item There is a partial function  $h_\beta : \Reals^\infty \times \Reals^\infty \times \Reals^{p}\times (0,\infty)\hookrightarrow [0,1]$ (called $f_\beta$'s \emph{limit extension}) such that $h_\beta$ has the following properties. 
\begin{itemize}
    \item $domain(h_\beta)$ is the set of all $(a,c,\bar u,\epsilon)\in \Reals^\infty \times \Reals^\infty \times \Reals^{p}\times (0,\infty)$ such that there is a non-decreasing sequence $\set{(a_i,c_i)}^\infty_{i=0} \in \Iag(\bar u)$ (i.e., $a_i\leq a_{i+1},c_i\leq c_{i+1}$) with    $\lim_{i\to\infty} (a_i,c_i) = (a,c).$ 
    \item For any non-decreasing sequence  $\set{(a_i,c_i)}^\infty_{i=0} \in \Iag(\bar u)$ such that     $\lim_{i\to\infty} (a_i,c_i) = (a,c)$,  $h_\beta(a,c,\bar u,\epsilon)=\lim_{i\to\infty} f_\beta (a_i,c_i,\bar u, \epsilon).$
    \item $h_\beta$  is  parametrically definable in $\threal.$
    
\end{itemize}


\end{enumerate}
Observe that $f_\beta$ and its limit extension $h_\beta$ agree on $domain(f_\beta)$. Therefore, a limit-definable region $\reg=(\Iag,f_\beta)$ shall be represented by a triple $(\Tag, \psi_h, \phi_h)$ where $(\psi_h,\phi_h)$ defines $h_\beta$ (the limit extension of $f_\beta$). 
\end{definition}
Intuitively, the first requirement  ensures  that the $\alpha$ that needs to be considered for a fixed $\gamma$ is rational. 
The second requirement ensures that $f_\beta$ is continuous ``from below" and can be extended to its boundary (including the case when $\alpha$ takes the value $\infty.$) Note that $\psi_h$ can be written in the theory of linear arithmetic thanks to the fact that $\Tag$ is a linear arithmetic formula.


\begin{example}
\label{example:limit-definable}
The region $\eta_1$ in Example~\ref{example:simple} can be seen to limit-definable with  $\Tag(x_\alpha,x_\gamma,\bar x)=((x_\alpha \geq 0)\wedge (x_\gamma=0))$ and $h_\beta$ as follows: 
$$h_\beta(a,c,\bar u,\epsilon) =\begin{cases}
                                 0 & \mbox{if } a=\infty, c=0, \bar u \in \cU,\epsilon>0\\
                                 \euler^{-\frac {a  \epsilon} 2} & \mbox{if }  a\in \Reals^{\geq 0}, c=0, \bar u \in \cU,\epsilon>0 \\
                                 \mbox{undefined} &\mbox{otherwise}.
                                \end{cases}$$
An example of a region that is simple but not limit-definable is the region $\eta_3$ defined as follows. 
Given $\bar u \in \cU$, $\epsilon>0,$
$\eta_3(\bar u, \epsilon)= \set{(a,b,0)\st \mbox{ either } (0\leq a< 1 \wedge  b= e^{-\frac {a \epsilon } 2}) \mbox{ or } (1\leq a \wedge b= e^{-\frac {3 a \epsilon } 5})}.
$ 
$\eta_3$ is not limit-definable as parameter $\beta$  has a ``discontinuity" at $\alpha=1$.
\end{example}

We have the following theorem that shows that checking Accuracy-at-an-input is decidable for fixed $\gamma.$  Please note that we require $\din$ to be definable in $\folreal$ to ensure that the distance to disagreement for a rational input is rational.  All examples considered in Section~\ref{sec:accuracy-ex} satisfy these constraints. In the decision procedure,  the fixed value of $\gamma$, $c$, is encoded in the formula $\Tag$. 
\ifdefined\AppendixTrue
The proof can be located in Appendix~\ref{app:fixedgamma}. 
\fi

\begin{theorem}  \label{thm:fixedgamma}
 The problem Accuracy-at-an-input is  decidable for {\newlang} programs $P_\epsilon$ and rational inputs $\bar u$ when
\begin{enumerate}
\item $\det(P_\epsilon), \din,\dout$ are definable in $\folreal,$ and
\item $\reg=(\Tag, \psi_h,\phi_h)$ is limit-definable and there is a rational number $c$ such that
for all  $a,c',\bar u$, if $\Tag(a,c',\bar u) $ is true then $c'=c.$


\end{enumerate}

\end{theorem}

\subsubsection*{Program with finite outputs} We now turn our attention to programs with finite outputs. For such programs, $\gamma$ is often $0.$ If that is the case, then we can appeal to Theorem~\ref{thm:fixedgamma} directly. However, for some examples, $\gamma$ may not be $0$ (for example, {\nmax} in Section~\ref{sec:nmax}).

When a program has only finite outputs, for each input $\bar u$, $\dout(\bar u,\bar v, \bar v')$ can take only a finite number of distinct values. This suggests that we need to check accuracy at input $u$ for only a finite number of possible values of $\gamma,$ namely the distinct values of $\dout(\bar u,\bar v, \bar v')$.  Then as in Theorem~\ref{thm:fixedgamma}, we can check for accuracy at these values of $\gamma$ by setting the $\alpha$ parameter to be $\dd(P_\epsilon,u)$, the distance to disagreement for $\bar u.$
We need a monotonicity condition that ensures the soundness of this strategy. 



\begin{definition}
\label{def:agmonotonic}
Let $\reg=(\Tag,f_\beta,h_\beta)$ be a limit-definable region. Given $\bar u$ and non-negative $a\in \Reals^\infty,c\in \Reals,$ let $I_{<a,c}(\bar u)$ be the set $\set{a' \st \Tag(a',c,\bar u) \mbox{ is true},
a'< a}.$ $\reg$ is said to be $(\alpha,\gamma)$-monotonic if
for each $\bar u$ and non-negative real numbers $c_1,c_2,a$ such that $I_{<a,c_1}(\bar u)\ne \emptyset$ and $I_{<a,c_2}(\bar u)\ne \emptyset$, $$
c_1\leq c_2 \Rightarrow \sup (I_{<a,c_1}(\bar u)) \leq \sup (I_{<a,c_2}(\bar u)) .$$
\end{definition}

We have the following result.
\ifdefined\AppendixTrue
The proof can be located in Appendix~\ref{app:finiteoutputs}.
\fi
\begin{theorem}
\label{thm:finiteoutputs}
The problem Accuracy-at-an-input is  decidable for {\newlang} programs $P_\epsilon$ and rational inputs $\bar u$ when (a) $P_\epsilon$ has finite outputs, (b) $\det(P_\epsilon),\din,\dout$ are definable in $\folreal,$ and (c) $\reg$ is $(\alpha,\gamma)$-monotonic.

\end{theorem}


%

When the  program $P_\epsilon$ has finite inputs and finite outputs, we can invoke Theorem~\ref{thm:finiteoutputs} repeatedly to check for accuracy at all possible inputs. 
The following is an immediate corollary of Theorem~\ref{thm:finiteoutputs}.

\begin{corollary}
\label{cor:finite}
The problem Accuracy-at-all-inputs is  decidable for {\newlang} programs $P_\epsilon$ when (a) $P_\epsilon$ has finite inputs and finite outputs, (b) $\det(P_\epsilon),\din,\dout$ are definable in $\folreal,$ and (c) $\reg$ is $(\alpha,\gamma)$-monotonic.

\end{corollary}




	\section{Experiments}
\label{sec:experiments}

We implemented a simplified version of the algorithm for verifying
accuracy of {\newlang} programs. Our tool {\toolplus} handles
loop-free programs with finite, discrete input domains, and whose
deterministic function has discrete output. Programs with bounded
loops (with constant bounds) can be handled by unrolling. The
restriction that the deterministic function has discrete output does
not preclude programs with real outputs, as they can be modeled in
the subset of {\newlang} that the tool handles. We discuss this
further below. 

The tool takes as input a program $P_\epsilon$
parametrized by $\epsilon$ and an input-output table representing
$\det(P_\epsilon)$ for a set of inputs, and either verifies $P_\epsilon$ to be
$(\alpha,\beta,\gamma)$-accurate for each given input and for all $\epsilon>0$ or returns a
counterexample, consisting of a specific input and a value for
$\epsilon$ at which accuracy fails. We choose values of
$\alpha,\beta,\gamma$ depending on the example. 
As accuracy claims in Section~\ref{sec:accuracy-ex} show,  $\beta$ is typically given as a continuous function of $\epsilon, \alpha$ and $\gamma$. For such continuous $\beta$, we can use $\alpha \leq \dd(P,u)$ in the definition of accuracy  instead of $\alpha < \dd(P,u)$ (see Definition~\ref{def:accuracy} on Page \pageref{def:accuracy}). In our experiments, we fix  $\gamma$ to be some integer.
For a given input $u$, $\alpha$ is usually set to be the distance to disagreement for the input being checked.   $\beta$ can thus be viewed as a function of $\epsilon$ and the input $u$.  The proof of Theorem~\ref{thm:fixedgamma} implies that such  checks are necessary and sufficient to conclude 
accuracy at the given inputs for fixed $\gamma,$ and all possible values of $\epsilon$ and $\alpha$. In many examples, the only value for $\gamma$ that needs to be verified is $0$.

 {\toolplus} is implemented in C++ and
uses Wolfram Mathematica\textregistered. It works in two phases. In
the first phase, a Mathematica script is produced with commands for the input-output probability computations and the subsequent inequality checks. In the second phase, the generated script is run on Mathematica.
We only verify accuracy-at-an-input and not accuracy-at-all-inputs as our decision procedure for the latter problem is subject to Schanuel's conjecture.

We test the ability of {\toolplus} to verify accuracy-at-an-input for four
examples from Section~\ref{sec:accuracy-ex}: $\svttwo$, $\nmax$,
Laplace Mechanism (denoted $\laplace$ below), and $\nsvt$. We also
verify a variant of $\svttwo$ which we refer to as $\svtone$ (the
difference is discussed in~\ref{sec:scaling} below). The pseudocode is shown in Figures~\ref{fig:nmax} and~\ref{fig:two-algos}
on Pages~\pageref{fig:nmax} and ~\pageref{fig:two-algos} respectively; we omit the
pseudocode for {\laplace} and refer the reader to its description in
Section~\ref{sec:accuracy-ex}. Three of the examples, $\svttwo$,
$\svtone$, and $\nmax$, have discrete output and thus their
deterministic functions, $\det(\svttwo),\det(\svtone)$, and
$\det(\nmax)$, are naturally modeled with a finite input-output
table. The other two examples, $\laplace$ and $\nsvt$, can be modeled
using a finite input-output table as follows. Given $\gamma$, we can
compute the deterministic function alongside the randomized function
and instrument the resulting program to check that all continuous
outputs are within the error tolerance given by $\gamma$, outputting
$\top$ if so. The finite input-output table can reflect this scheme by
regarding the instrumented program as a computation which outputs
$\top$.

Our experiments test two claims. We first test
the performance of {\toolplus} by measuring how running time scales
with increasing input sizes and example parameters. Running times are
given in Tables~\ref{tab:perf-svt-nmax-lap}
and~\ref{tab:num-sparse-experiment}. Here the parameter $m$ is the length of input arrays and a range $[-\ell,\ell]$ is the range of all possible integer values that can be stored in each array location. Thus, we have $(2\ell+1)^m$ possible inputs. Hence, the tool behaves roughly polynomial in $\ell$ and exponential in $m$. In the second part of our
experiments, we show that {\toolplus} is able to obtain accuracy bounds that are better than those known in
the literature, and generate counterexamples when accuracy claims are not true. All experiments were run on an Intel\textregistered Core
i7-6700HQ @ 2.6GHz CPU with 16GB memory. In the tables, running times
are reported as (T1/T2), where T1 refers to the time needed by the C++
phase to generate the Mathematica scripts and T2 refers to the time
used by Mathematica to check the scripts. In some tables we omit T1
when it is negligible compared to T2.

We note the following about the experimental results.
\begin{enumerate}
\item {\toolplus} verifies accuracy in reasonable time. The time
  needed to generate Mathematica scripts is significantly smaller than
  the time taken by Mathematica to check the scripts (i.e., T1 $\ll$
  T2). Most of the time spent by Mathematica goes toward computing
  output probabilities.
\item {\toolplus} is able to verify that accuracy holds with smaller
  error probabilities than those known in the literature.
\item Verifying accuracy is faster than verifying differential
  privacy. Differential privacy involves computing, for any specific input, a
  probability for each possible output, whereas in accuracy we only
  need to compute the single probability for each input-output pair
  given by the deterministic function.
\end{enumerate}

\subsection{Performance}
\label{sec:scaling}
Table~\ref{tab:svt-experiment1} shows running times for $\svttwo$ and
$\svtone$. $\svtone$ differs from $\svttwo$ by only sampling a single
noisy threshold, whereas $\svttwo$ samples a fresh noisy threshold
each time it finds a query above the current noisy
threshold; the pseudocode for $\svtone$ is the same as $\svttwo$ (Figure~\ref{fig:nmax}), except for the re-sampling of $\rv_T$ inside the for-loop (shown in bold). Generally, running time increases in the number of inputs
that have to be verified. For both $\svttwo$ and $\svtone$, we verify
$(\alpha,\beta,0)$-accuracy for all $\alpha,\beta$, with
$\beta=2mce^{-\nicefrac{\alpha\epsilon}{8c}}$. This can be
accomplished with single accuracy checks at $\alpha=\dd(\svttwo,u)$
and $\alpha=\dd(\svtone,u)$, for each input $u$, as discussed in the
proof of Theorem~\ref{thm:fixedgamma}. Observe there is a substantial
performance difference between $\svttwo$ and $\svtone$ for $c>1.$ When analyzing
$\svtone$, {\toolplus} must keep track of many possible relationships
between random variables and the single noisy threshold. This becomes
expensive for many queries and large value of $c$. On the other hand,
whenever $\svttwo$ samples a new noisy threshold, this has the effect
of decoupling the relationship between future queries and past
queries. Table~\ref{tab:noisy-max-experiment1} shows results for
$\nmax$, in which we verify $(\alpha,\beta,0)$-accuracy for all
$\alpha,\beta$ with $\beta=me^{-\nicefrac{\alpha\epsilon}{2}}$. Here
again, a single accuracy check at $\alpha=\dd(\nmax,u)$ suffices for
each input $u$. Table~\ref{tab:laplace-experiment1} shows results for
$\laplace$, where we verify $(0,\beta,\gamma)$-accuracy for all
$\beta,\gamma$, with
$\beta=ke^{-\nicefrac{\gamma}{\epsilon}}$. Finally,
Table~\ref{tab:num-sparse-experiment} shows results for $\nsvt$, where
we verify $(\alpha,\beta,\alpha)$-accuracy for specific values of
$\alpha$, with $\beta=(2m+1)ce^{-\nicefrac{\alpha\epsilon}{9c}}$,
i.e. the improved error probability bound from
Section~\ref{sec:accuracy-ex}. In this case, the accuracy claim holds for
precisely $(\alpha,\beta,\alpha)$. This is because $\gamma=\alpha$ is
no longer constant, and we thus cannot leverage
Theorem~\ref{thm:fixedgamma} for a stronger claim.
Observe that when the input range is $[-\ell,\ell]$, then  $\dd(\nsvt,x)$ can take any integer value between $0$ and $\ell.$
Thus, even though we have set $\alpha$ to be the same value in each individual experiment (and not the distance to disagreement for the input), we vary $\alpha$ across the experiments so as to ensure that for each input $x$, accuracy is eventually verified when $\alpha$ is set to $\dd(\nsvt,u)$.

\begin{table*}[!t]
\caption{\footnotesize (a) Running times for $\svttwo$ and $\svtone$,
  verifying $(\alpha,\beta,0)$-accuracy for all $\alpha,\beta$, where
  $\beta=2mce^{-\nicefrac{\alpha\epsilon}{8c}}$. Accuracy is verified
  for all lists of $m$ integer-valued queries ranging over
  $[-1,1]$. The threshold $T$ is set to be $0$. (b) Running times for $\nmax$ with $m=3$ and varying input
  range, verifying $(\alpha,\beta,0)$-accuracy for all $\alpha,\beta$,
  with $\beta=me^{-\nicefrac{\alpha\epsilon}{2}}$. (c) Running times
  for $\laplace$, verifying $(0,\beta,\gamma)$-accuracy with $k=2$,
  $\Delta=1$, and $\beta=ke^{-\gamma\epsilon}$.}
\label{tab:perf-svt-nmax-lap}
  \begin{minipage}[t]{\textwidth}\centering
  \begin{subtable}[t]{\textwidth}\centering
  \resizebox{\textwidth}{!}{%
  \begin{tabular}[t]{| l | c | c | c | c | c | c | c | c | c | c |}
    \hline
    $m$ & 1 & 2 & 2 & 3 & 3 & 3 & 4 & 4 & 4 & 4\\ \hline
    $c$ & 1 & 1 & 2 & 1 & 2 & 3 & 1 & 2 & 3 & 4\\ \hline
    $\svttwo$ & 0s/12s & 0s/45s & 0s/45s & 0s/97s & 0s/90s & 0s/89s & 0s/195s & 1s/189s & 2s/189s & 1s/195s\\ \hline
    $\svtone$ & 0s/11s & 0s/44s & 0s/72s & 0s/99s & 0s/217s & 0s/386s & 1s/199s & 1s/462s & 1s/940s & 1s/1467\\
    \hline
  \end{tabular}
  }
  \caption{\footnotesize $\svttwo$ and $\svtone$}
  \label{tab:svt-experiment1}
  \end{subtable}
  \end{minipage}
\begin{minipage}[t]{\textwidth}
  \begin{subtable}[t]{0.55\textwidth}\centering
  \resizebox{.7\textwidth}{!}{
  \begin{tabular}[t]{| l | c | c | c |}
    \hline
    Range & $[-1,1]$ & $[-2,2]$ & $[-3,3]$ \\ \hline
    (T1/T2) & 0s/148s & 1s/823s & 1s/2583s \\ \hline
  \end{tabular}
  }
  \caption{\footnotesize $\nmax$}
  \label{tab:noisy-max-experiment1}
\end{subtable}
\hfill
\begin{subtable}[t]{0.35\textwidth}\centering
  \resizebox{0.6\textwidth}{!}{
  \begin{tabular}[t]{| l | c | c | c |}
    \hline
    $\gamma$ & 1 & 2 & 3 \\ \hline
    $[-1,1]$ & 5s & 6s & 6s \\ \hline
    $[-2,2]$ & 9s & 9s & 8s \\ \hline
  \end{tabular}
  }
  \caption{\footnotesize $\laplace$}
  \label{tab:laplace-experiment1}
\end{subtable}
\end{minipage}
\end{table*}

\begin{table*}[!t]
 \caption{\footnotesize Running times for $\nsvt$, verifying
    $(\alpha,\beta,\alpha)$-accuracy for
    $\beta=(2m+1)ce^{-\nicefrac{\alpha\epsilon}{9c}}$ at specific
    values of $\alpha,m,c$, and varying input ranges. Threshold $T = 0$ in both tables. 
    The tables highlight scaling in run time as $m$ and $\alpha$ vary.
    We use the improved expression for
    $\beta$ from Corollary~\ref{cor:NumSparse}. }
  \label{tab:num-sparse-experiment}
  \begin{subtable}[t]{\textwidth}\centering
  \resizebox{\textwidth}{!}{%
  \begin{tabular}[t]{| l | c | c | c | c | c | c | c | c | c | c | c | c |}
    \hline
    $m$ & 1 & 1 & 1 & 2 & 2 & 2 & 3 & 3 & 3 & 4 & 4 & 4 \\ \hline
    $\alpha$ & 1 & 1 & 2 & 1 & 1 & 2 & 1 & 1 & 2 & 1 & 1 & 2 \\ \hline
    Range & [-1,1] & [-2,2] & [-2,2] & [-1,1] & [-2,2] & [-2,2] & [-1,1] & [-2,2] & [-2,2] & [-1,1] & [-2,2] & [-2,2]
    \\ \hline
    (T1/T2) & 0s/18s & 0s/36s & 0s/19s & 0s/65s & 1s/287s & 1s/68s & 0s/178s & 2s/1669s & 1s/181s & 2s/332s & 
    20s/6065s & 7s/350s \\ \hline
  \end{tabular}
  }
  \caption{\footnotesize $\nsvt$ with $c = 1$}
  \end{subtable}
  \begin{subtable}[t]{\textwidth}\centering
  \resizebox{0.8\textwidth}{!}{%
  \begin{tabular}[t]{| l | c | c | c | c | c | c | c | c | c |}
    \hline
    $m$ & 2 & 2 & 2 & 3 & 3 & 3 & 4 & 4 & 4 \\ \hline
    $\alpha$ & 1 & 1 & 2 & 1 & 1 & 2 & 1 & 1 & 2 \\ \hline
    Range & [-1,1] & [-2,2] & [-2,2] & [-1,1] & [-2,2] & [-2,2] & [-1,1] & [-2,2] & [-2,2]\\ \hline
    (T1/T2) & 0s/65s & 0s/296s & 0s/66s & 0s/169s & 2s/1537s & 1s/170s & 1s/311s & 18s/5783s &3s/315s \\ \hline
  \end{tabular}
  }
  \caption{\footnotesize $\nsvt$ with $c = 2$}
  \end{subtable}
\end{table*}

\subsection{Improved Accuracy Bounds and Counterexamples}
\label{sec:better-bounds}
In
Tables~\ref{tab:better-bounds-svt},~\ref{tab:better-bounds-noisy-max},
and~\ref{tab:better-bounds-numeric}, we summarize the results of
verifying better accuracy bounds than those known in the literature
for $\svttwo$, $\svtone$, $\nmax$, and $\nsvt$. For each algorithm and
parameter choice, we display the best accuracy bound obtained by
searching for the largest integer $\kappa$ for which {\toolplus} was
able to verify $(\alpha,\nicefrac{\beta}{\kappa},\gamma)$-accuracy,
where $\alpha,\beta,\gamma$ are the corresponding values from the
experiments in Subsection~\ref{sec:scaling}. In each row, we include a
counterexample returned by the tool after a failed attempt to verify
$(\alpha,\nicefrac{\beta}{\kappa+1},\gamma)$-accuracy. Each
counterexample consists of both an $\epsilon$ and a specific input $u$
at which the accuracy check failed.

\begin{table*}[!t]
\caption{Accuracy is checked for all vectors of length $m$ with
  entries ranging over [-1,1]. Best column displays
  $\nicefrac{1}{\kappa}$, where $\kappa$ is the largest integer for which
  the tool verified accuracy. The counterexample column displays
  $\nicefrac{1}{\kappa+1}$ along with $\epsilon$ and input $u$ at
  which the accuracy check failed. We include $\beta$ for clarity
  because of its dependence on $c$. }
  \centering
  \begin{subtable}[t]{\textwidth}\centering
  \resizebox{0.8\textwidth}{!}{
  \begin{tabular}[t]{| c | c | c | c | c | c | c |}
    \hline
    Algorithm & $c$ & $\beta$ & Best & (T1/T2) & Counterexample & (T1/T2) \\ \hline
    $\svttwo$ & 1 & $6e^{-\nicefrac{\alpha\epsilon}{8}}$ & $\nicefrac{1}{6}$ & 0s/101s & $\nicefrac{1}{7}$, $u=[-1,-1,1]$, $\epsilon=\nicefrac{17}{10}$ & 0s/110s  \\ 

    $\svttwo$ & 2 & $12e^{-\nicefrac{\alpha\epsilon}{16}}$ & $\nicefrac{1}{12}$ & 0s/91s & $\nicefrac{1}{13}$, $u=[1,-1,1]$, $\epsilon=\nicefrac{50}{19}$ & 0s/95s  \\ \hline

    $\svtone$ & 1 & $6e^{-\nicefrac{\alpha\epsilon}{8}}$ & $\nicefrac{1}{6}$ & 0s/99s & $\nicefrac{1}{7}$, $u=[-1,-1,0]$, $\epsilon=\nicefrac{67}{106}$ & 0s/111s  \\ 

    $\svtone$ & 2 & $12e^{-\nicefrac{\alpha\epsilon}{16}}$ & $\nicefrac{1}{13}$ & 0s/220s & $\nicefrac{1}{14}$, $u=[-1,-1,-1]$, $\epsilon=\nicefrac{17}{13}$ & 0s/216s  \\ \hline

  \end{tabular}
  }
  \caption{\footnotesize Improved accuracy bounds for $\svttwo$ and
    $\svtone$, with
    $m=3$. 
  }
  \label{tab:better-bounds-svt}
\end{subtable}
\begin{subtable}[t]{\textwidth}\centering
  \resizebox{0.65\textwidth}{!}{
  \begin{tabular}[t]{| c | c | c | c | c | c |}
    \hline
    $m$ & $\beta$ & Best & (T1/T2) & Counterexample & (T1/T2) \\ \hline
    3 & $3e^{-\nicefrac{\alpha\epsilon}{2}}$ & $\nicefrac{1}{4}$ & 0s/154s & $\nicefrac{1}{5}$, $u=[-1,0,0]$, $\epsilon=\nicefrac{27}{82}$ & 0s/176s\\ 


    4 & $4e^{-\nicefrac{\alpha\epsilon}{2}}$ & $\nicefrac{1}{4}$ & 0s/874s & $\nicefrac{1}{5}$, $u=[0,0,1,0]$, $\epsilon=\nicefrac{52}{23}$ & 0s/910s\\ \hline

  \end{tabular}
  }
  \caption{\footnotesize Improved accuracy bounds for $\nmax$.
  }
  \label{tab:better-bounds-noisy-max}
\end{subtable}
\begin{subtable}[t]{\textwidth}\centering
  \resizebox{0.65\textwidth}{!}{
  \begin{tabular}[t]{| c | c | c | c | c | c | c |}
    \hline
    $\alpha$ & $c$ & $\beta$ & Best & (T1/T2) & Counterexample & (T1/T2) \\ \hline


    1 & 1 & $7e^{\nicefrac{-\epsilon}{9}}$ & $\nicefrac{1}{3}$ & 1s/174s & $\nicefrac{1}{4}$, $u=[-1,-1,1]$,
                                         $\epsilon=37$ & 1s/173s \\ 

    1 & 2 & $14e^{\nicefrac{-\epsilon}{18}}$ & $\nicefrac{1}{5}$ & 1s/162s & $\nicefrac{1}{6}$, $u=[-1,1,1]$,
                                         $\epsilon=59$ & 1s/167s \\ \hline
  \end{tabular}
  }
  \caption{\footnotesize Improved accuracy bounds for $\nsvt$, with
    $m=3$. 
  }
  \label{tab:better-bounds-numeric}
\end{subtable}

\end{table*}

	\section{Related work}
\paragraph{Accuracy proofs}
The Union Bound logic of \citet{BGGHS16-icalp} is a program logic for
upper bounding errors in probabilistic computation. 
It is a lightweight probabilistic program logic: assertions are
predicates on states, and probabilities are only tracked through an
index that accounts for the cumulative error. The Union bound logic
has been used to prove accuracy bounds for many of the algorithms
considered in this paper. However, the proofs must be constructed
manually, often at considerable cost. Moreover, all reasoning about
errors use union bounds, so precise bounds that use concentration
inequalities are out of scope of the logic. Finally, the proof system
is sound, but incomplete. For instance, the proof system cannot deal
with arbitrary loops. In contrast, our language allows for
arbitrary loops and we provide a decision procedure for accuracy.
Trace Abstraction Modulo Probability (TAMP) in \citet{SmithHA19}, is an
automated proof technique for accuracy of probabilistic programs. TAMP
generalizes the trace abstraction technique of \citet{HeizmannHP09} to the
probabilistic setting. TAMP follows the same lightweight strategy as
the union bound logic, and uses failure automata to separate between
logical and probabilistic reasoning. TAMP has been used for proving
accuracy of many algorithms considered in this paper. However, TAMP
suffers from similar limitations as the Union Bound logic (except of
course automation): it is sound but incomplete, and cannot deal with
arbitrary loops and concentration inequalities.

\paragraph{Privacy proofs}
There is a lot of work on verification and testing of privacy
guarantees~\cite{RP10,BKOZ13,GHHNP13,ZK17,AH18,BichselGDTV18,DingWWZK18,BartheCJS020}. We
refer to~\cite{BGHP16} for an overview.

\paragraph{Program analysis}
There is a large body of work that lifts to the probabilistic setting
classic program analysis and program verification techniques,
including deductive verification~\cite{Kozen85,Morgan96,Kaminski19},
model-checking~\cite{Katoen16,kwiatkowska2010advances}, abstract
interpretation~\cite{Monniaux00,CousotM12}, and static program
analysis~\cite{SankaranarayananCG13,WangHR18}. Some of these
approaches rely on advanced techniques from probability theory,
including concentration
inequalities~\cite{Sankaranaranyanan__2019__Concentration} and
martingales~\cite{ChakarovS13,ChatterjeeFNH16,BartheEFH16,KuraUH19,WangHR20}
for better precision.

These techniques can compute sound upper bounds of the error
probability for a general class of probabilistic programs. However,
this does not suffice to make them immediately applicable to our
setting, since our definition of accuracy involves the notion of {\distance}. Moreover, these works generally do not address the
specific challenge of reasoning in the theory of reals with
exponentials. Finally, these techniques cannot be used to prove the violation of accuracy claims; our approach can both prove and disprove accuracy claims.

\paragraph{Hyperproperties}
Hyperproperties~\cite{ClarksonS10} are a generalization of program
properties and encompass many properties of interest, particularly in
the realm of security and privacy. Our definition of accuracy falls in
the class of 3-properties, as it uses two executions of $det(P)$ for
defining distance to disagreement, and a third execution of $P$ for
quantifying the error.

There is a large body of work on verifying hyperproperties. While the
bulk of this literature is in a deterministic setting, there is a
growing number of logics and model-checking algorithms from
probabilistic hyperproperties~\cite{AbrahamB18,WangNBP19,DimitrovaFT20}.
To our best knowledge, these algorithms do not perform parametrized
verification, and cannot prove accuracy for all possible values of
$\epsilon$.

\section{Conclusions}
\label{sec:conclusions}
We have introduced a new uniform definition of accuracy, called
$(\alpha,\beta,\gamma)$-accuracy, for differential privacy
algorithms. This definition adds parameter $\alpha$, that
accounts for {\distance}, to the traditional parameters $\beta$ and
$\gamma$. This uniform, generalized definition can be used to unify
under a common scheme different accuracy definitions used in the literature that quantify the probability of getting approximately correct answers,
including \emph{ad hoc} definitions for classical algorithms such as
{\athres}, {\svttwo}, {\nsvt}, {\nmax} and others. Using the
$(\alpha,\beta,\gamma)$ framework of accuracy we were able to improve
the accuracy results for \nsvt.  We have shown that checking
$(\alpha,\beta,\gamma)$-accuracy is decidable for a non-trivial class
of algorithms with a finite number of real  inputs and outputs, that are
parametrized by privacy parameter $\epsilon$, described in
our expanded programming language \newlang, for all values of
$\epsilon$ within a given interval $\:I,$ assuming Schanuel's conjecture.
We have also shown that the problem of checking accuracy at a single input is
decidable under reasonable assumptions without assuming  Schanuel's
conjecture for programs in {\newlang}, even when the inputs and outputs
can take real values. This implies that checking accuracy is decidable for programs whose inputs and outputs take values in finite domains.
Finally, we presented experimental results
 implementing our approach by adapting {\tool} to check accuracy at specified inputs for
 {\athres}, {\svttwo}, {\nsvt} and {\nmax}.


In the future, it would be interesting to study how our decision
procedure could be used for automatically proving concentration
bounds, and how it could be integrated in existing frameworks for
accuracy of general-purpose probabilistic computations. It would also
be interesting to extend our results and the results
from \citet{BartheCJS020} to accommodate unbounded number of inputs and outputs, and other probability
distributions, e.g. Gaussian mechanism. On a more practical side, it
would be interesting to study the applicability of our techniques in
the context of the accuracy first approach from \citet{LigettNRWW17}.

\begin{acks}
The authors would like to thank anonymous reviewers for their interesting and useful comments. 
This work was partially supported by \grantsponsor{NSF}{National Science Foundation}{https://www.nsf.gov/} grants \grantnum{NSF}{NSF CNS 1553548}, \grantnum{NSF}{NSF CCF 1900924}, \grantnum{NSF}{NSF CCF 1901069} and \grantnum{NSF}{NSF CCF 2007428}. 
\end{acks}
     \ifdefined\AppendixTrue
     
     \else \newpage
     \fi
	\bibliography{header,main}
    
    \ifdefined\AppendixTrue
	\appendix
        \section{Proofs of Results from Section~\ref{sec:accuracy-ex}}
\label{sec:acc-ex-proofs}

We present the missing proofs from Section~\ref{sec:accuracy-ex}.

\subsection{Proof of Lemma~\ref{lem:nmax}}
 
Consider any input $u= (q_1,\ldots q_m)$. Now, we
consider two cases. The first case is when all elements in $u$ have
equal values. In this case, $\det(\mbox{\nmax})$ outputs $1$, the index
of the first element. It is easy to see that
$dd(\mbox{\nmax},u)=0$ and {\nmax} is trivially
$(\alpha,\beta,0)$-accurate at $u$ for all $\alpha\geq 0.$ The second
case, is when there are at least two elements in $u$ whose values are
unequal. Let $\s{max}$ be the maximum value of the elements in $x$ and
$\s{smax}$ be the maximum of all values other than the maximum
value in $u.$ It is easy to show that
$dd(\mbox{\nmax},u)\:=\frac{1}{2}(\s{max} - \s{smax}).$ Clearly
$dd(\mbox{\nmax},u)>0.$ Now it is enough to show that
{\nmax} is $(\alpha,\beta,0)$-accurate at $u$ for $\beta
\:=m\euler^{-\frac{\alpha\epsilon}{4}}$ and for all $\alpha<
\dd(\mbox{\nmax},u)\:=\frac{1}{2}(\s{max} - \s{smax}).$ Now consider
any $\alpha<\frac{1}{2}(\s{max} - \s{smax}).$ Let $j$ be the output of $\det(\mbox{\nmax})$ on
input $u$, i.e., $j$ is the smallest index of an element in $u$ whose
value is $\s{max}.$ Let $i$ be an index output by {\nmax}. 
From the
above mentioned result of \cite{BGGHS16-icalp}, we see that 
$Pr (q_j-q_i <
\frac{4}{\epsilon}\ln\frac{m}{\beta})\geq 1-\beta.$
Substituting $\beta\:=m\euler^{-\frac{\alpha\epsilon}{2}}$, after
simplification, we see
that $Pr (q_j-q_i <
\frac{4}{\epsilon}\ln\frac{m}{\beta})= Pr(q_j-q_i<2\alpha) =Pr(q_i=q_j)$, since
$2\alpha<\s{max} - \s{smax}.$ Now, $Pr(q_i=q_j)=Pr(i\in  B(j,u,0))\geq 1-\beta.$
Hence {\nmax} is $(\alpha,\beta,0)$-accurate.

\subsection{Proof of Theorem~\ref{thm:nsvt}}

	Consider any output $v=(v_1,...,v_k)\in \cV$ such  that
  $v_k\in \Reals$ and there are exactly $c$ values of $i$ such that
  $v_i\in \Reals.$ Let $I\:=\set{i\st v_i\in \Reals}.$ Let
  $J=\set{i\notin I\st 1\leq i\leq m}.$ Let $U'\:=
  (\det(NumericSparse))^{-1}(v).$ Fix $\alpha\in \Reals$ such that $\alpha>0.$
Consider any $u=(u_1,...,u_m)\in U' $ such that
$d(u,\:\cU-U')>\alpha.$ Now, consider any $i$ such that $i\leq k$ and $i\in J.$ Let
$w=(w_1,..,w_m)$ such that $w_i=T$ and for $j\neq i$, $w_j=u_j.$
Clearly $w\in \cU-U'$ and $d(u,w)>\alpha.$ Hence we see tat $u_i<T-\alpha.$
Now, consider any $i\in I$. For any $\delta>0$, let
$w_\delta=(w_1,...,w_m)$ such that $w_i = T-\delta$ and for $j\neq i$,
$w_j=u_j.$ Clearly $w_\delta \in \cU-U'$ and $d(u,w)>\alpha.$ Hence,
$u_i-w_i>\alpha$ and hence $u_i>T+\alpha-\delta.$ The last inequality
holds for every $\delta>0$ and hence $u_i\geq T+\alpha$. The above
argument shows that, $\forall i\in I,\:u_i\geq T+\alpha$ and $\forall
i\in J,i\leq k,\:u_i<T-\alpha.$ Now using the accuracy result for Sparse, we
see that {\nsvt}, on input $u$, generates an output
$o=(o_1,...,o_k)$ where $\forall i\in I,\:o_i\in \Reals$ and $\forall
j\leq k, j\in J,$ $o_i=\bot$, with probability $\geq 1-\beta'_1$ where
$\beta'_1$ is given by the equation
$$\alpha\:= \frac{9c}{\epsilon} (\ln k + \ln (\frac{2c}{\beta'_1}))$$
 Now, the above equation gives us,
 $\beta'_1\:=2kc\euler^{-\frac{\alpha\epsilon}{9c}}.$
Now, taking $\beta_1= 2mc\euler^{-\frac{\alpha\epsilon}{9c}}$, we see
that $\beta_1\geq \beta1'$ and {\nsvt} outputs an output
sequence of the form $o$ with probability $\geq 1-\beta_1.$ Now, using
accuracy result for the Laplacian mechanism, we see that with
probability $\geq 1-\beta_2$,  it is the case that
$\forall i\in I,\:|o_i-u_i|\leq \gamma$, where $\beta_2$ is given by
the equation
$$\gamma\:= \ln(\frac{c}{\beta_2})\frac{9c}{\epsilon}$$
Now, the above equation gives us $\beta_2=c\euler^{-\frac{\gamma
    \epsilon}{9c}}.$
Thus, we see that, on input $u$, {\nsvt} outputs an sequence in
$O=\set{o=(o_1,...,o_k)\st \forall i\in I,|o_i-u_i|\leq \gamma,\: \forall
  i\in J\:o_i=\bot}$ with probability $\geq (1-\beta_1)(1-\beta_2)\geq
1-(\beta_1+\beta_2).$ Taking $\beta=\beta_1+\beta_2$, we get that with
$$\beta = 2mc\euler^{-\frac{\alpha\epsilon}{9c}} + c\euler^{-\frac{\gamma
    \epsilon}{9c}}.$$  {\nsvt} outputs a sequence in $O$ on
input $u$ with probability $\geq 1-\beta.$

        \section{Proof of Theorem~\ref{thm:undec}}
\label{sec:undec-proof}

  We show the undecidability result for the Accuracy-at-an-input
  problem. Essentially, we reduce the problem of checking whether a
  2-counter machine terminates when started with both counter values
  being initially zero. Given a description of a 2-counter machine
  $M$, our reduction constructs a randomized program $P$ (which is
  independent of $\epsilon$) such that $P$ satisfies the accuracy
  condition iff $M$ does not halt when started with zero counter
  values. The reduction is similar to the one used in the proof for
  undecidability of checking differential privacy for randomized
  programs given in \citet{BartheCJS020}.

  Let $M$ be the given 2-counter machine. The program $P$ has five
  variables: $in,out,counter1,counter2$, and $state$. Here, $in,out$
  are the input and output variables, respectively. The variables
  $counter1,counter2$, and $state$ are used to capture the two counter
  values and the state of $M$, respectively. Initially
  $out,counter1,counter2$ are set to $0$ and $state$ is set to the
  initial state of $M$. If $in=1$ then $P$ sets $out$ to $1$ and
  terminates. If $in=0$ then P repeats the following step. If the
  variable $state$ denotes the final control state of $M$ then $P$
  sets $out$ to $1$ and terminates; otherwise, with probability
  $\frac{1}{2}$, $P$ sets $state$ to the final control state of $M$,
  and with probability $\frac{1}{2}$ it simulates one step of $M$
  using the values of $state, counter1$, and $counter2$. Since the
  value of $out$ at termination is the output value, it is easy to see
  that on input $0$ (i.e., initial value of $in$ is zero), the output
  value is $0$ with probability $1$ iff $M$ does not terminate when
  started with zero values in both counters. The deterministic program
  $\det(P)$ outputs $0,1$ on inputs $0$ and $1$, respectively.  Using
  our notation, $\cU\:=\cV\:=\set{0,1}.$ The distance function $d$ on
  $\cU$ is given by $d(0,1)=1$. The distance function $d'_u$ (for all
  input values $u$) on $\cV$ is given by $d'_u(0,1)=1.$ Now, it is
  easy to see that $P$ is $(0,0,0)$-accurate at input $0$ iff $M$ does
  not halt when started with zero counter values.
        
\section{Proof of Theorem~\ref{thm:condAcc}}
\label{app:condAcc}
The following observation shall be used in our result.
\begin{lemma}
\label{lem:lraisdip}
Let $\phi(\bar x)$ be a formula in the first-order theory of linear real arithmetic with $n$ free variables $\bar x.$ Then there is a deterministic {\newlang} program $P_\phi$ that has $n$ input  real variables and one output $\sdom$ variable such that   $P_\phi$ outputs $1$ on input $\bar a$ if $\phi(\bar a)$ is true over the set of reals and outputs $0$ otherwise. Furthermore, 
$P_\phi$ can be taken to be a program independent of $\epsilon.$
\end{lemma}
\begin{proof}
The first order theory of reals enjoys quantifier-elimination. Hence, $\phi(\bar x)$ can be written as a quantifier-free formula with $n$ variables in conjunctive normal form. From this observation, it is easy to see that there is a  deterministic {\newlang} program $P_\phi$ that evaluates the truth of $\phi.$ 
\end{proof}

Recall that the distance to disagreement for an input $u$, $\dd(u)$ for the deterministic function $\det(P_\epsilon)$ is  the distance of the input $u$ to the set  of inputs $u'$ which output a different value that $u$ when $\det(P_\epsilon)$ is applied. We generalize this definition as follows.     

\begin{definition}
Let $f:\cU \rightarrow \cV$ be a function and let $d$ be a distance function on $\cU.$ Given $u \in \cU,$ let $\ddf{f}{d}:\cU \rightarrow \Reals^\infty$ be the function such that 
$\ddf{f}{d}(u)=\inf \set{d(u,u') \st f(u)\ne f(u')}.$
\end{definition}

We have the following observation.  
\begin{proposition}
\label{prop:diam}
Let $P_\epsilon$ be a {\newlang} program implementing the deterministic function $\det(P).$ Let $\din$ be an input distance function. If
$\det(P)$ and $\din$ are definable  in $\foreal$ then $\ddf{\det(P_\epsilon)}{\din}$ is definable in $\foreal.$ Further, if $\det(P)$ and $\din$ are definable  in $\folreal$
then so is $\ddf{\det(P_\epsilon)}{\din}.$ 
\end{proposition}

We describe the decision procedure. Let $(\psi_\det,\phi_\det)$ be the formulas defining $\det(P_\epsilon),$ $(\psin,\phin)$ be the formulas defining $\din$, $(\psiout,\phiout)$ be the formulas defining $\dout$, and $(\psi_R,\phi_R)$ be the formulas defining $\reg.$

Let $P_\epsilon$ have $p$ input variables and $q$ output variables.
Let $\cU$ be the set of all valuations over the input variables and $\cV$ be the set of all valuations over the output variables. From the fact that $\dout$ is definable in $\folreal,$
it can be shown that there is a formula $\phi_{\leq\gamma}(\bar x, \bar y, \bar z, x_\gamma)$ in linear arithmetic with $p+2\,q+1$ variables such that $\phi_{\leq\gamma}(\bar u, \bar v, \bar{v}', c)$ holds iff $\bar u \in \cU, \bar v, \bar{v}' \in \cV$ and $\dout(\bar u, \bar v, \bar{v}') \leq c.$  From Lemma~\ref{lem:lraisdip}, it follows that there is a {\newlang} program $P_{\leq \gamma}$ with $p+2 q+1$ input variables that evaluates $\phi_{\leq\gamma}.$ 

Now, we construct a new {\newlang} program $P^{new}_\epsilon$ from $P_\epsilon.$  Let $\vrin$ and $\vrout$ represent the input and output variables of the program $P_\epsilon$.   $P^{new}_\epsilon$ has $p+q+1$ input variables, $\vrin,\ovr,$ and $\rv_\gamma$.  Here $\vrin$ are the input variables of $P_\epsilon,$  $\ovr$ are $q$ fresh variables and
are of the same type as the output variables of $P_\epsilon,$ and $\rv_\gamma$ is a fresh real variable.

Intuitively, $P^{new}_\epsilon$ will simulate $P_\epsilon$ and output $1$ if the output of the program $P_\epsilon$
is at most a distance $\rv_\gamma$ away from $\ovr$ and $0$ otherwise. $P^{new}_\epsilon$ proceeds as follows:
\begin{enumerate}
\item stores $\vrin$ and $\ovr$ into fresh real variables $\rvin$ and $\orv$ respectively,
\item runs $P_\epsilon,$
\item stores the output $\vrout$ in $\rvout,$
\item runs $P_{\leq\gamma}$ on $\rvin,\rvout,\ovr,\rv_\gamma,$ and
\item the output of $P_{\leq\gamma}$ is taken to be the output of  $P^{new}_\epsilon.$
\end{enumerate}
Now, consider the partial function $\Pr_{1,P^{new}_\epsilon}$ that maps an input of  $P^{new}_\epsilon$ to the probability of  $P^{new}_\epsilon$ outputting $1$ as in Lemma~\ref{lem:ProbDef}. Let $(\psipr,\phipr)$ be the formulas defining $\Pr_{1,P^{new}_\epsilon}$ as given by Lemma~\ref{lem:ProbDef}.

Let $\bar x$ be variables representing inputs to $\det(P_\epsilon)$ and let $\bar y$ represent the output of $\det(P_\epsilon)$ on input $\bar x.$
Let $(\psi_\dd,\phi_\dd)$ be the formulas defining the function $\ddf {\det(P_\epsilon)} {\din}$ as given by Proposition~\ref{prop:diam}.
We makes the following observations.
\begin{itemize}
\item the formula $\phi_\det(\bar x, \bar y)$ asserts that $\bar y$ is the output of $\det(P_\epsilon)$ on input $\bar x,$
\item the formula $\phi_\dd(\bar x, 1,0)$ asserts that the distance to disagreement for $\bar x$   is $\infty,$
\item the formula $\phi_\alpha (\bar x, x_\alpha) \equiv \phi_\dd(\bar x, 1,0) \vee \exists x_d.\ (\phi_\dd(\bar x, 0, x_d) \wedge (x_d>x_\alpha)),$
asserts that the distance to disagreement for $\bar x$ is $>x_\alpha.$

\item $\phipr(\bar x,\bar y, x_\gamma,\epsilon,x_{pr})$ asserts that the probability of  $P_\epsilon$ outputting a value on input $\bar x$ that is at most $x_\gamma$ away from $\bar y$ away is $x_{pr},$ and
\item $(x_{pr} \geq 1-x_\beta)$ asserts that $x_{pr}$ is greater than equal to $x_\beta.$
\end{itemize}
It is easy to see that the formula
\begin{align}
Acc(\bar x,\epsilon)&\equiv  \forall x_\alpha.\ \forall  x_\beta.\  \forall  x_\gamma.\ \forall  \bar y. \ \nonumber \\
                 & \qquad (\psi_R(x_\alpha,x_\beta,x_\gamma,\bar x,\epsilon) \wedge \phi_\alpha (\bar x, x_\alpha) \wedge \phi_\det(\bar x, \bar y)) \nonumber \\
                  & \qquad \qquad \rightarrow \forall x_{pr}.\  (\phipr(\bar x,\bar y, x_\gamma,\epsilon,x_{pr}) \rightarrow (x_{pr} \geq 1-x_\beta)) \label{eq:accx}
\end{align}
asserts the accuracy of $P_\epsilon$ at input $\bar x$ (for $\epsilon$.)
Thus, $P_\epsilon$ is  accurate if and only if the sentence
   $$ \forall \epsilon.\ \forall \bar x.\ (((\epsilon>0) \wedge \psi_\det(\bar x))  \rightarrow  Acc(\bar x,\epsilon))$$
   is true over the theory of reals with exponentiation.
Thus, the result follows.

\section{Proof of Corollary~\ref{cor:accpoint}}
Let  $Acc(\bar x,\epsilon)$ be the sentence constructed as in the proof of Theorem~\ref{thm:condAcc} as given in Section~\ref{app:condAcc}.  $P_\epsilon$ is  accurate if and only if the formula
   $$ \forall \epsilon. \ (((\epsilon>0) \wedge \psi_\det(\bar u))  \rightarrow  Acc(\bar u,\epsilon))$$ is true. The result follows.

\section{Proof of Theorem~\ref{thm:infpoint}}
\label{app:infpoint}
Let $f_\beta$ be defined using the formulas $(\psi_\beta, \phi_\beta).$
Let $\phi_\det(\bar x, \bar y), \phi_\alpha (\bar x, x_\alpha), \phipr(\bar x,\bar y, x_\gamma,\epsilon,x_{pr}),$ be as in the Proof of Theorem~\ref{thm:condAcc}. Observe that 
as $\phi_\det(\bar x, \bar y)$ is a linear arithmetic formula,  the formula $\phi_\det(\bar u, \bar y)$ is also a linear arithmetic formula. From the fact that $\phi_\det(\bar u, \bar y)$ is true for only one value of $\bar y$ and the fact that linear arithmetic enjoys quantifier-elimination, we can compute the unique vector $\bar v$ of rational numbers such that $\phi_\det(\bar u, \bar v)$ is true; note that, the fact that $\phi_\det$ is a formula in linear arithmetic ensures that $\bar{v}$ is rational. By definition,
$\bar v$ is the output of function $\det(P_\epsilon)$ on $\bar u$.

Now, accuracy of $P_\epsilon$ at $\bar u$  for $\epsilon$ is given by the formula:
\begin{align}
Acc_{\bar u}(\epsilon)&\equiv  \forall  x_\beta.\ 
                  (\phi_\beta (a,c,\bar u, \epsilon, x_\beta) \wedge \phi_\alpha (\bar u, a)) \nonumber \\
                  & \qquad  \rightarrow \forall x_{pr}.\  (\phipr(\bar u,\bar v, c,\epsilon,x_{pr}) \rightarrow (x_{pr} \geq 1-x_\beta)) \label{eq:accu}
\end{align}
Thus, $P_\epsilon$ is  accurate at $\bar u$ if and only if the formula 
   $ \forall \epsilon.\ \forall \bar x.\ (\epsilon>0) \wedge  Acc_{\bar u}(\epsilon)$ is true. Using the fact that $\bar u, \bar v, a, c,$ are rational, we see that  $\phipr(\bar u,\bar v, c,\epsilon,x_{pr})$,  $\phi_\beta (a,c,\bar u,\epsilon, x_\beta)$ are in $\lngreal.$ Hence, we can check that $P_\epsilon$ is  accurate at $\bar u$ or not. 

\section{Proof of Theorem~\ref{thm:fixedgamma}}
\label{app:fixedgamma}
We first make the following observation. 

\begin{proposition}
\label{prop:supcomp}
Let $\emptyset \ne S\subseteq \Reals^\infty$ be a  definable set in $\folreal$ when viewed as a subset of $\Reals^2$. Then either $\sup(S)$ is $\infty$ or a rational number.  Further, $\sup(S)$ can be computed.
\end{proposition}

We proceed with the proof.  
Thanks to the fact that $\det(P_\epsilon)$ and $\din$ are definable in $\folreal,$ we have  that the function $\dd_{\det(P_\epsilon),\din}$ is definable in $\folreal$ (See Proposition~\ref{prop:diam}).  Thanks to Proposition~\ref{prop:supcomp}, we can conclude that the distance to disagreement of $\bar u$
is either $\infty$ or a rational number, and can be computed. Let us denote this distance by $\dd(\bar u).$ Observe that the definition of accuracy implies that we need to check accuracy for all $a< \dd(\bar u).$

Now consider the set,  $S_c= \set{a\st a\in \Reals^{>0}, a < \dd(\bar u), (a,c) \in \Iag(u)}.$ Clearly  $S_c$ is definable in $\folreal.$ If $S_c$ is empty, $P_\epsilon$ is trivially accurate at $\bar u.$ Note that emptiness of $S_c$ can be checked as $S_c$ is definable in $\folreal.$

Now assume $S_c$ is non-empty. $a^\mx_c=\sup(S_c)$ is either $\infty$ or a rational number and can be computed as a consequence of Proposition~\ref{prop:supcomp}. 

Let $\bar v$ be the output of $\det (P_\epsilon).$ As in Theorem~\ref{thm:infpoint}, $\bar v$ can  be computed. Let $P^{new}_\epsilon$ be the program as constructed in the proof of Theorem~\ref{thm:condAcc}.
It is easy to see that the limit-definability of  $\reg$ and the anti-monotonicity of $f_\beta$  implies that $(a^\mx_c,c,\bar u,\epsilon)\in domain(h_\beta)$ for each $\epsilon>0$ and that checking accuracy at input $\bar u$ is equivalent to checking if $P^{new}_\epsilon$ outputs $1$ on input $\bar u,\bar v, c$ with probability $\geq 1-h_\beta(a^\mx_c,c,\bar u,\epsilon).$ Thanks to Lemma~\ref{lem:ProbDef} and the parametric definability of $h_\beta$, this can be expressed as a formula in $\lngreal.$





\section{Proof of Theorem~\ref{thm:finiteoutputs}}
\label{app:finiteoutputs}

We describe the decision procedure. Let  $P_\epsilon$ have $p$ inputs and $q$ outputs. 
Let $\bar v$ be the output of $\det(P_\epsilon)$ on $\bar u.$ As in the Proof of Theorem~\ref{thm:infpoint}, $\bar v$ can be computed. 
As in the proof of Theorem~\ref{thm:fixedgamma}, the distance to disagreement for $\bar u$, can be computed.  For the rest of the proof we denote this by $\dd(\bar u).$
Observe also that since the first-order theory of linear arithmetic enjoys quantifier-elimination,  $\dout(\bar u,\bar v,\bar{v}')$ is either $\infty$ or a rational number for each $\bar{v}'$,  which can be computed (See Proposition~\ref{prop:supcomp}).

Now, let $0=c_0<c_1\cdots < c_j$ be the distinct numbers in $\Reals^\infty$ such that $c_i = \dout(\bar u,\bar v,\bar{v}'_i)$ for some $\bar{v}'_i.$ Let $S_0\subsetneq S_1\cdots \subsetneq S_j=\cV$ be sets such that 
for each $i$, $S_i$ is the set of outputs which are at most distance $c_i$ away from $\bar v.$ We observe that for any $c\geq 0$, the ball $B(\bar v, \bar u, c)$ must be one of the sets $S_1,\ldots, S_j.$ For each $i,$ let $p_i(\epsilon)$ denote the probability of  $P_\epsilon$ outputting an element in the set $S_i$ (as a function of $\epsilon$). $p_i$ is parametrically definable in $\threal$ as a consequence of Lemma~\ref{lem:ProbDef}. Note that $p_j(\epsilon)$ is $1$ since we assume all {\newlang} programs terminate with probability $1.$

Let $O$ be the set of all indices $\ell$ such that there   exists  $(a,c)\in \Iag(\bar u)$ with $c_{\ell}\leq c<c_{\ell+1}$ and  $a<\dd(\bar u).$   
Observe that this set can be computed thanks to the fact that  $\Iag$ is $\folreal$ definable. 

 Assume that $O$ is non-empty. For $\ell \in O,$
let $Z_\ell=\set{(a,c)\st a<\dd(\bar u), c_\ell\leq c < c_{\ell+1}, (a,c) \in \Iag}.$ Let $c^\mx_\ell$ be $\sup \set{c \st \exists a.  (a,c) \in Z_\ell}.$  Let $a^\mx_\ell = \sup \set{a \st \exists c.  (a,c) \in Z_\ell}.$ Note that $c^\mx_\ell$ and $a^\mx_\ell$ can be computed to Proposition~\ref{prop:supcomp}.
\begin{claim}
\label{claim:monotonicty}
We have the following for each $\ell\in O$.
\begin{enumerate}
    \item There is a  non-decreasing sequence $\set{(a_k,c_k)}^\infty_{i=0}$ such that $(a_k,c_k)\in Z_\ell$ and 
          $\lim_{k\to\infty} (a_k,c_k) = (a^\mx_\ell,c^\mx_\ell).$
    \item 
    $(a^\mx_\ell,c^\mx_\ell,\bar u,\epsilon) \in domain (h_\beta)$ for each $\epsilon>0,$
    \item For each $(a,c)\in Z_\ell$ and $\epsilon>0,$ $f(a,c,\bar u,\epsilon)\geq h(a^\mx_\ell,c^\mx_\ell,\bar u,\epsilon).$
\end{enumerate}
\end{claim}
\begin{proof} (Proof of claim)
Fix $(a,c)\in Z_\ell.$ Let $a_0=a, c_0=c.$ Consider the set $\Gamma=\set{c'\st c'\geq c_0, c^\mx_\ell -\frac 1 2 \leq c' \leq c^\mx_\ell \mbox{ and } \exists a'. (a',c')\in Z_\ell}.$

Assume first $a<a^\mx_\ell.$ Note that $\Gamma$ is non-empty thanks to the definition of $c^\mx_\ell.$ Now, we claim that there must be a $(a_1,c_1) \in Z_\ell$ such that 
$c_1\in \Gamma, a_0 \leq a_{1}$ and $a^\mx_\ell - \frac 1 2 \leq a_{1} \leq a^\mx_\ell.$ If no such $a_1,c_1$ exists then by definition, we have that if $(a',c') \in Z_\ell$  for $c'\in \Gamma$ then $a'< \max(a_0,a^\mx_\ell - \frac 1 2).$ Now, it is easy to see that $(\alpha,\gamma)$-monotonicity of $\reg$ implies that whenever $(a',c')\in Z_\ell$ then $a'< \max(a_0,a^\mx_\ell - \frac 1 2).$ This, in turn,  implies that 
$a^\mx_\ell \leq \max(a,a^\mx_\ell - \frac 1 2)$ which is a contradiction.

Hence, there is $(a_1,c_1)\in Z_\ell$ such that 
$a_0 \leq a_{1}, a^\mx_\ell - \frac 1 2 \leq a_{1} \leq a^\mx_\ell, c_0 \leq c_{1}$ and $c^\mx_\ell - \frac 1 2 \leq c_{1} \leq c^\mx_\ell. $  
Proceeding in a similar fashion, we can show by induction, that for each $i,$ there exist $(a_{i+1},c_{i+1})\in Z_\ell$ such that $a_i \leq a_{i+1}, a^\mx_\ell - \frac 1 {i+2}\leq a_{i+1} \leq a^\mx_\ell, c_i \leq c_{i+1}$ and $c^\mx_\ell - \frac 1 {i+2} \leq c_{i+1} \leq c^\mx_\ell. $  Note that 
$\set{(a_i,c_i)}^\infty_{i=0}$ is the desired sequence in Part 1 of the claim. The other two parts are an easy consequence of limit-definability and anti-monotonicty of $f_\beta$ respectively. 

Now, assume $a=a^\mx_\ell.$ Thanks to $(\alpha,\gamma)$-monotonicity and the definition of $a^\mx_\ell$, it is easy to see that $(a^\mx_\ell,c') \in Z_\ell$ whenever $c'\in \Gamma.$ Further, we can construct a sequence $c_0<c_1<c_2,\ldots$ such that $c_i\in \Gamma$ and $\lim_{i\to\infty} c_i = c^\mx_\ell.$ The first part of the claim will now follow by taking $a_i=a^\mx_\ell$ for each $i.$  The other two parts an easy consequence of limit-definability and anti-monotonicty of $f_\beta$ respectively.
\end{proof}
From limit-definability of $\reg$ and Claim~\ref{claim:monotonicty}, it follows that if $P_\epsilon$ is accurate at $\bar u$ then for each  $\ell\in O$  and $\epsilon>0,$  
 \begin{align}
  p_\ell(\epsilon)\geq 1- h_\beta(a^\mx_\ell,c^\mx_\ell,\bar u, \epsilon). \label{eqn:max}
 \end{align} 
Observe that this check can be performed for each $\ell\in O$ since these checks can be encoded as a sentence in $\lngreal.$ Hence, the Theorem will follow if we can show that it suffices to check Equation~\ref{eqn:max} for all $\ell\in O.$ 

Assume that Equation~\ref{eqn:max} holds for all $\ell\in O.$ Fix $a,c \in \Reals^{\geq 0}$ such that  $(a,c) \in \Iag(\bar u).$ 
If $c\geq c_j$, $P_\epsilon$ is  $(a,f_\beta(a,c,\epsilon),c)$-accurate trivially since the probability of obtaining an output in the set $S_j$ is $1.$ 

Otherwise let $k$ be the unique $ \ell $ such $c_\ell \leq c < c_{\ell+1}.$ If $k $ is not in the set $O$ then  $P_\epsilon$ is $(a,f_\beta(a,c,\epsilon),c)$-accurate trivially as $a\geq \dd(\bar u)$. Otherwise note 
that by Claim~\ref{claim:monotonicty}, 
 $f_\beta(a,c,\epsilon) \geq h_\beta(a^\mx_k,c^\mx_k,\bar u, \epsilon).$ $(a,f_\beta(a,c,\epsilon),c)$-accuracy now follows from the fact that the check given by Equation~\ref{eqn:max} succeeds for $\ell=k$. The concludes the proof of the Theorem.

        \fi
        
\end{document}